\documentclass[a4paper]{article} 

\usepackage{amsfonts}
\usepackage{amsmath} 
\usepackage{amssymb}
\usepackage{graphicx}
\usepackage{amsthm}
\usepackage{dsfont}    
\usepackage{microtype} 

\usepackage{caption}
\usepackage{subcaption}
\usepackage{color}
\usepackage{multirow}
\usepackage{url}

\usepackage[utf8]{inputenc}
\usepackage[english]{babel}

\usepackage{filecontents}
 
\newtheorem{theorem}{Theorem}[section]

\newtheorem{lemma}[theorem]{Lemma}

\newtheorem{proposition}[theorem]{Proposition}

\theoremstyle{definition}
\newtheorem{definition}{Definition}[section]
 
\theoremstyle{remark}

\newcommand{\onehalf}{\frac{1}{2}}
\newcommand{\N}{\mathbb{N}}
\newcommand{\R}{\mathbb{R}}
\newcommand{\M}{\mathrm{M}}

\newcommand{\G}{\mathcal{G}}
\newcommand{\E}{\mathcal{E}}

\newcommand{\SO}{\mathrm{SO}(3)}

\newcommand{\var}{\mathrm{Var}}

\newcommand{\ra}{\rightarrow}

\newcommand{\fatone}{\mathds{1}}
\newcommand{\dx}{\ \mathrm{d}x}
\newcommand{\dy}{\ \mathrm{d}y}

\newcommand{\dt}{\ \mathrm{d}t}

\newcommand{\dm}{\ \mathrm{d}\mu}

\begin{document}
	\title{Approximation to uniform distribution in  $\SO$}
	\author{Carlos Beltrán and Damir Ferizovi\'{c}\thanks{The first author was supported by the Spanish ``Ministerio de
			Econom\'ia y Competitividad'' under projects MTM2017-83816-P and
			MTM2017-90682-REDT (Red ALAMA), as well as by the Banco Santander and
			Universidad de Cantabria under project 21.SI01.64658.
			The second named author thankfully acknowledges support
	by the Austrian Science Fund (FWF): F5503 ``Quasi-Monte Carlo Methods'' and by the NAWI Graz Funding.  }}
\maketitle
\begin{abstract}
Using the theory of determinantal point processes we give upper bounds for the Green and Riesz energies for the rotation group $\SO$, with Riesz parameter up to 3.
The Green function is computed explicitly, and a lower bound for the Green energy is established, enabling comparison of uniform point constructions on $\SO$.
 The variance of rotation matrices sampled by the determinantal point process is estimated, and formulas for the $L^2$-norm of Gegenbauer polynomials with index 2 are deduced, which might be of independent interest. Also a simple but effective algorithm to sample points in $\SO$ is given.
 
\end{abstract}

\section{Introduction and Results}
In this paper we study properties of a finite collection of randomly generated points in $\SO$, the rotation group of 3-dimensional Euclidean space, sampled by \emph{determinantal point processes} (dpp). It turns out that they tend to be well distributed, a property that is important for discretization, integration and approximation. Our goal is not to compute actual collections of evenly distributed rotation matrices, but rather to provide a comparison tool that allows to decide the effectiveness of any given method.

If one is given an algorithm to generate finite (but arbitrarily large) collections of matrices, common methods to measure how well distributed these are, include either calculating some discrete energy of them or looking at the speed of convergence of the counting measure towards uniform measure. Most work in this direction has been done on spheres of various dimensions, see for instance \cite{Beltran}, etc.; the particular question of finding collections of points with very small energy was posed by Shub and Smale in \cite{Shub} and is nowadays known as Smale's 7th problem \cite{Smale}.

In order to extend part of the work done on spheres to the context of rotation matrices, we will obtain bounds on various energies for points generated through the method of dpp, which are technically speaking counting measures where one identifies them with their set of atoms. In few words, such a process is obtained by taking a Hilbert space $\mathcal{H}(X)$ of an underlying measure space $(X,\mu)$ and an $N$-dimensional subspace $V\subset\mathcal{H}(X)$, with projection kernel $\mathcal{K}$ onto $V$ -- then, under mild conditions on $X$, one is guaranteed  almost surely the existence of such a process with $N$ distinct points in $X$ associated to $\mathcal{K}$. 

The theory of those processes has been developed in \cite{GAF}; there one also finds a pseudo-code which  samples points based on the dpp -- which seems hard to implement. A main feature of the underlying points is that they tend to ``repel'' each other, and hence have become the theoretical basis of construction of well-distributed points on various symmetric spaces, see for instance \cite{ASZ14,Etayo,Beltran,Marzo}. 

Since one can sometimes compute the expected value of the energy of points coming from these processes with high precision, they have been used as a tool to understand the asymptotic properties of the discrete energy in that context; and in particular, for even dimensional spheres with exception of the usual $2$-sphere, the best known bounds for some energies have been proved using this approach.


We will employ the same method for $\SO$, considering first the (discrete) {\em  Riesz $s$-energy} for $A=\{\alpha_1,\ldots,\alpha_N\}\subset\SO$: 
$$E_{R}^s(A):=\sum_{j\neq k} \frac{1}{\|\alpha_j-\alpha_k\|_F^s},$$
with $\alpha_j$ being thought of as rotation matrices, $\|\cdot\|_F$ being the Frobenius or $L^2$-norm, and $s\in(0,3]$. In contrast to this, the \emph{continuous Riesz $s$-energy} is given by replacing the double sum by the double integral over $\SO$.
We further set 
$$\mathcal{E}_{R}^s(N)=\inf_{|A|=N}E_{R}^s(A) .$$ 

The investigation of these sums is very popular and results describe the behavior of the two leading terms. This seems particularly interesting in case $s$ equals the dimension, where we have following result.
\begin{theorem}\label{thm_Riesz3}
Let $N=\binom{2L+3}{3}$ for $L\in\N$, then 
the Riesz 3-energy satisfies
$$ 12\sqrt{2}\pi\cdot\mathcal{E}_{R}^3(N)\ \leq\ N^2\log(N)+
\big(3\gamma+\log(8^2\cdot6)-\tfrac{21}{4}\big)N^2+O(N^{5/3}\log(N)).$$
\end{theorem}
The right-hand side is the expected value of the Riesz 3-energy with underlying points generated by a dpp. Now, given any particular method of generating finite point sets in $\SO$, one can compute,  numerically, their $3$-energy and compare it to the value above to decide if the points are evenly distributed. This comparison would clearly rise in significance at the presence of lower bounds on the $3$-energy, which do not seem easy to find. For this reason we turn our attention to the Green energy, where we succeeded in this endeavor.

To recap, a Green function $\mathcal{G}_L$ for a linear differential operator $L$ is an integral kernel to produce solutions for inhomogeneous differential equations and is unique modulo kern($L$). In our case, we deal with the Laplace-Beltrami operator $\Delta_g$, and note that kern($\Delta_g$) is the set of harmonic functions -- which are just constants on a compact Riemannian manifold ($\M,g$).
We will  construct $\mathcal{G}=\mathcal{G}_{\Delta_g}$ in such a way, that it integrates to zero and speak of \emph{the} Green function.

The (discrete) {\em Green energy} for $A=\{\alpha_1,\ldots,\alpha_N\}\subset\SO$ will be given by
$$E_{\G}(A):=\sum_{i\neq j}\G(\alpha_i,\alpha_j), $$
and we let 
$$\mathcal{E}_{\G}(N)=\inf_{|A|=N}E_{\G}(A).$$
It is noteworthy that $\G(\alpha,\beta)\cdot d(\alpha,\beta)\approx1$ for $\alpha$ close to $\beta$ in geodesic distance $d(\cdot,\cdot)$, and a set of points with small Green energy is hence expected to be well-distributed, which is indeed the main result in \cite{Juan}: We know that if  $\{\alpha_1,\ldots,\alpha_N\}$ attains the minimal possible energy, then the associated discrete measure approaches the uniform distribution in $\SO$ as $N\to\infty$.
A set of points with small Green energy is also expected to be well-separated, see \cite{Criado}.

Now, $\G(\cdot,\beta)$ is for any $\beta\in\SO$ a zero mean function by definition, and if  $\alpha_1,\ldots,\alpha_N$ were simply chosen uniformly and independently in $\SO$, then the expected value of the Green energy would equal $0$, so in particular we have $\mathcal{E}_{\G}(N)\leq0$. In this note we prove the following much stronger result.
\begin{theorem}\label{thm_Green} Let $N=\binom{2L+3}{3}$ for $L\in\N$, then
	\begin{equation*}
	-3\sqrt[3]{\pi}N^{4/3}+
O(N)\leq\ \mathcal{E}_{\G}(N)\ \leq -4\left(\tfrac{3}{4}\right)^{4/3}N^{4/3}+
O(N).
	\end{equation*}
\end{theorem}

The right-hand side is the expected value of the Green energy with underlying points generated by a  dpp, and that is where we have the restriction for $N$, as the process is related to subspaces $V$ that we can project onto. The lower bound is valid for all $N$.

As mentioned above, another classical measure of the distribution properties of $\alpha_1,\ldots,\alpha_N$ is the speed of convergence to uniform measure, i.e. choosing some range sets $\{A_j\}_{j\in I}$ measurable w.r.t. Haar measure $\mu$ and investigating the behavior of 
$$\sup_{j\in I}\Big| \#\{k\ :\ \alpha_k\in A_j\} -N\cdot\mu(A_j)\Big|$$
as $N$ grows large. We will tackle this problem probabilistically, where we turn the count of points in $A_j$ into a random variable.

In analogy to spherical caps on spheres, the range sets for $\SO$ will be chosen to be balls $B(\alpha,\varepsilon):=\{\beta\in\SO\ :\ \omega(\alpha^{-1}\beta)<\varepsilon\}$ for $\varepsilon\in\big(0,\frac{\pi}{2}\big)$ and $\omega(\cdot)$ being the rotation angle distance introduced in the following sections. For given random points $\alpha_1,\ldots,\alpha_N$, we define random variables via characteristic functions
$$X_{\alpha,\varepsilon}^k=\chi_{B(\alpha,\varepsilon)}(\alpha_k)\hspace{0.5cm} \mbox{ and }\hspace{0.5cm} \eta_{\alpha,\varepsilon}=\sum_{k=1}^N X_{\alpha,\varepsilon}^k.$$
Now, for a collection of random uniform points, chosen independently in $\SO$ we have
$$\mathbb{E}[\eta_{\alpha,\varepsilon}]=N\mu(B(\alpha,\varepsilon))=N\mu(B(\fatone,\varepsilon)),$$
and the variance can also be computed from the independence of the points:
$$
\var[\eta_{\alpha,\varepsilon}]=\mathbb{E}[\eta_{\alpha,\varepsilon}^2]-\mathbb{E}[\eta_{\alpha,\varepsilon}]^2=N\left(\mu(B(\fatone,\varepsilon))-\mu(B(\fatone,\varepsilon))^2\right).
$$
We are able to bound the variance of this quantity for our dpp, proving that it is much smaller than in the previous case.
\begin{theorem}\label{thm_variance}
  Let $N=\binom{2L+3}{3}$ for $L\in\N$, and $\varepsilon\in\big(0,\frac{\pi}{2}\big)$ be fixed, then the points generated by our determinantal point process satisfy
  $$
  \mathbb{E}[\eta_{\alpha,\varepsilon}]=N\mu(B(\alpha,\varepsilon))=N\mu(B(\fatone,\varepsilon)),
  $$
  and moreover 
 $$\var(\eta_{\alpha,\varepsilon})= O\Big(\frac{\varepsilon^2}{\cos(\varepsilon)}\Big)\cdot N^{2/3}\log(N).$$
\end{theorem}
From Theorem \ref{thm_variance} and for any fixed $\varepsilon$, we then
have by Chebyshev's inequality
$$\sup_{\alpha\in\SO}\mathbb{P}\Big(\big|\eta_{\alpha,\varepsilon}-N\mu(B(\fatone,\varepsilon))\big|\geq T\Big)\leq\var(\eta_{\alpha,\varepsilon})T^{-2}; $$
for example, letting $T=N^{1/3}\log(N)$ and with some little arithmetic we obtain
$$\sup_{\alpha\in\SO}\mathbb{P}\Big(\big|\tfrac{1}{N}\eta_{\alpha,\varepsilon}-\mu(B(\fatone,\varepsilon))\big|\geq \tfrac{\log (N)}{N^{2/3}}\Big)= O\big(\tfrac1{\log(N)}\big). $$
In other words, for large $N$ the counting and Haar measures are very similar with large probability.

\section{Introductory Concepts}
In this section we collect some definitions and previous results that we will use and that intend to make this manuscript reasonably self-contained. Proofs and definitions of Chebyshev polynomials and alike are postponed to subsection \ref{subsec_proofs}.

\subsection{Structure, distances and integration in $\SO$}
The special orthogonal group $\SO$ is the compact Lie group of 3 by 3 orthogonal matrices over $\R$ that represent rotations in $\R^3$, i.e.
with determinant equal to one. Its exponential map is given by Rodrigues' rotation formula, and a closed expression for the Baker-Campbell-Hausdorff formula has been derived in \cite{Engo}. It is a $3$ dimensional manifold and since it is naturally included in $\R^9$ it is customary to let it inherit its Riemannian submanifold structure.

Following \cite{Schmid}, using Euler angles $(\varphi_1,\theta,\varphi_2)\in[0,2\pi)\times[0,\pi]\times[0,2\pi)$,
every element $R\in \SO$ can be decomposed as $R=s_z(\varphi_1)s_x(\theta)s_z(\varphi_2)$ where
\begin{equation*}
 s_z(\varphi_1):=\left( {\begin{array}{ccc}
   \cos(\varphi_1) & -\sin(\varphi_1) & 0 \\
   \sin(\varphi_1) & \cos(\varphi_1) & 0  \\
   0 & 0 & 1  \\
  \end{array} } \right),\
   s_x(\theta):=\left( {\begin{array}{ccc}
      1 & 0 & 0  \\
   0& \cos(\theta) & -\sin(\theta)  \\
  0&\sin(\theta) & \cos(\theta)  \\
  \end{array} } \right)
\end{equation*}
are rotations around the $z$-axis and $x$-axis respectively. The normalized Haar measure (i.e. the unique left and right invariant probability measure in $\SO$) is given 
by $\dm(R)=\frac{1}{8\pi^2}\sin(\theta)\mathrm{d}\varphi_1\mathrm{d}\theta\mathrm{d}\varphi_2$, and it corresponds to the inherited Riemannian submanifold structure of $\SO$ up to the normalizing constant.

The Riemannian distance associated to the structure of $\SO$ is certainly a natural and 
useful concept, but for us it will be more convenient to use the so called 
{\em rotation angle distance} defined as follows: for $\alpha,\beta\in\SO$,
$$\omega(\alpha^{-1}\beta)=
\arccos\left(\frac{\mbox{\textbf{Trace}}(\alpha^{-1}\beta)-1}{2}\right).$$


	Its convenience stems from following fact, see for example \cite[page 173]{Schmid}: Given a function $f\in L^1(\SO)$ such that we can find $\tilde{f}\in L^1([0,\pi])$
	with $f(x)=\tilde{f}(\omega(x))$, then
	\begin{equation}\label{eq_integralSO}
	\int_{\SO}f(x)\ \mathrm{d}\mu(x)=\frac{2}{\pi}\int_0^{\pi} \tilde{f}(t) \sin^2\big(\tfrac{t}{2}\big)\ \mathrm{d}t.
	\end{equation}

\subsection{Laplace-Beltrami operator and Green function in $\SO$}
The Laplace-Beltrami operator $\Delta_g$ is defined on any  Riemannian manifold $(\M,g)$ in terms of the Levi-Civita connection. Following \cite{DoCarmo}, if $\gamma_1(t),\ldots,\gamma_n(t)$ is a set of geodesics in an $n$-dimensional manifold such that $\gamma_j(0)=p\in\M$ for all $1\leq j\leq n$, and such that $\{\dot\gamma_j(0)\}$ 
form an orthonormal basis of the tangent space $T_p\M$ (geodesic normal coordinates), then the action of $\Delta_g$ on $C^2$-functions $f$ at $p$ is given by 
\[
\Delta_g f(p)=-\sum_{j=1}^n \frac{d^2}{dt^2}\Big|_{t=0}f(\gamma_j(t)).
\]
Note the convention given by the minus sign in front of the sum, which sometimes leads to confusion given  the Laplacian in $\R^n$. The convention we use here is widely accepted, see for example \cite{Jost}. A Green function $\G=\G_{\Delta_g}$ is a distributional solution to
\[
\Delta_g\G(\cdot,y)=\delta(\cdot,y)-\frac{1}{\mu_{dV}(\M)}.
\]
This way defined it is unique modulo kern($\Delta_g$) and it is common practice to add a constant in such a way that for all $y\in\M$ the function $\G(\cdot,y)$ has zero mean, see \cite{Aubin}. We use this convention and simply refer to $\G$ as {\em the} Green function.

It further follows from classical Fredholm theory for a linear differential operator $L$ that
\begin{equation}\label{eq:Greendef}
\mathcal{G}_L(x,y)=\sum_{j=1}^{\infty}\frac{\phi_j(x)\bar\phi_j(y)}{\lambda_j},
\end{equation}
where $0=\lambda_0<\lambda_1\leq\lambda_2\leq\cdots$ is the sequence of eigenvalues for $L$ and $\{\phi_j\}$, $j\geq1$ is a complete orthonormal set of associated eigenfunctions. This is hence true 
locally on any manifold, and the expression we obtain will  be independent of any particular chart, thus valid globally.
In the case $\M=\SO$, geodesics are dealt with in \cite{Novelia}, the eigenvalues and eigenfunctions of $\Delta_g$ are known from the classical theory of continuous groups and have been intensively studied in the physics literature, see \cite{Schmid,Joshi}, \cite[§15]{Wigner}:

\begin{lemma}\label{lem:eigens}
	The eigenvalues of $\Delta_g$ in $\SO$ are $\lambda_\ell=\ell(\ell+1)$ for $\ell\geq0$. 
	Moreover, if $H_\ell$ is the eigenspace associated to $\lambda_\ell$, then the dimension of
	$H_\ell$ is $(2\ell+1)^2$ and an orthonormal basis of $H_\ell$ is given by $\sqrt{2\ell+1}D_{m,n}^\ell$ where $-\ell\leq m,n\leq \ell$ and $D_{m,n}^\ell$ are Wigner's $D$-functions.
\end{lemma}
%
Moreover we have, see \cite[Eq. 4.65]{Joshi} or \cite[pp. 40-41]{Vollrath} for a nice summary:
\begin{equation}\label{eq_chebychevU}
\sum_{m=-\ell}^\ell\sum_{n=-\ell}^\ell \mathcal{D}_{m,n}^\ell(\alpha)
\overline{\mathcal{D}_{m,n}^\ell(\beta)}=\mathcal{U}_{2\ell}
\Big(\cos\big(\tfrac{\omega(\alpha^{-1}\beta)}{2}\big)\Big);
\end{equation}
where $\mathcal{U}_{2\ell}(x)$ is the Chebyshev polynomial of second kind and degree $2\ell$. The following simple form for the Green function is derived, and to the best of our knowledge, this is the first time it has been formulated.
\begin{lemma}\label{lem:Green}
	The Green function for the Laplace-Beltrami operator on $\SO$ can be written in terms of the 
	metric $\omega$, i.e. for $\alpha,\beta\in\SO$ with $\alpha\neq\beta$:
	\begin{equation*}
	\mathcal{G}(\alpha, \beta)=\big(\pi-\omega(\alpha^{-1}\beta)\big)
	\cot\big(\tfrac{\omega(\alpha^{-1}\beta)}{2}\big)-1.
	\end{equation*}
\end{lemma}


\subsection{Determinantal point processes}
We point the reader to the excellent monograph \cite{GAF} for an introduction to point processes, and we briefly summarize part of this material below. As in \cite{Beltran} and \cite{Etayo}, we will use only a fraction of the theory.

 A \emph{simple point process} on a locally compact Polish space 
 $\Lambda$ with reference measure $\mu$ is a random, integer-valued positive Radon measure $\eta$, that almost surely assigns at most measure 1 to singletons -- we shall think of it as a counting measure
 $$\eta=\sum_{j=1}\delta_{x_j},$$
 with $x_j\neq x_s$ for $j\neq s$. One usually 
 identifies $\eta$ with a discrete subset of $\Lambda$.
 
 The \emph{joint intensities} of $\eta$ w.r.t. $\mu$, if they exist, 
 are functions $\rho_k:\Lambda^k\ra [0,\infty)$ for 
 $k>0$, such that for pairwise disjoint sets $\{D_s\}_{s=1}^k\subset \Lambda$, the expected value of the product of number of points falling into $D_s$ is given by
 $$\mathbb{E}\bigg[\prod_{s=1}^k\eta(D_s)\bigg]=
 \int_{D_1\times\ldots\times D_k}\rho_k(y_1,\ldots,y_k)\dm(y_1)\ldots \dm(y_k),$$
and $\rho_k(y_1,\ldots,y_k)=0$ in case $y_j=y_s$ for some $j\neq s$.

A simple point process is \emph{determinantal}
 with kernel $\mathcal{K}$\footnote{ If $\mathcal{K}$ is a projection kernel, one  ought to say
 	\emph{determinantal projection process}.}, iff for every $k\in\N$ 
 and all $y_j$'s
 $$ \rho_k(y_1,\ldots,y_k)=\mbox{det}\Big(\mathcal{K}(y_j,y_s)\Big)_{1\leq j,s\leq k}.$$

 Let $(\M,g)$ be a compact Riemannian manifold with measure $d\mu(x)$. Let $H\subseteq L^2(\M)$ be any $N$-dimensional subspace in the set of square-integrable functions. It follows from the Macchi-Soshnikov theorem \cite[Thm. 4.5.5]{GAF} that a simple point process with $N$ points exists in $\M$ associated to $H$. Its main property is given by \cite[Form. (1.2.2)]{GAF}: For any measurable function $f:\M\times \M\to[0,\infty]$
  \begin{equation}\label{eq_expected_value_dpp}
 \mathbb{E}\bigg[\sum_{i\neq j} f(x_i,x_j)\bigg]=\iint_{\M}f(x,y)\Big(\mathcal{K}_{H}(x,x)\mathcal{K}_{H}(y,y)-|\mathcal{K}_{H}(x,y)|^2 \Big)\dm(x,y);
 \end{equation}
 where
 \begin{description}
 	\item $\mathbb{E}\big[g(x_1,\ldots,x_N)\big]$ means expected value of some function defined from $\M\times\cdots\times\M$ ($N$ copies of $\M$) to $[0,\infty]$, when $x_1,\ldots,x_N$ are chosen from the point process associated to $H$;
 	\item $\mathcal{K}_{H}(x,y)$ is the (orthogonal) projection kernel on $H$, namely for any $f\in L^2(\M)$ the orthogonal projection of $f$ onto $H$ satisfies:
 	\[
 	\Pi_H(f)(x)=\int_{y\in\M}f(y)\mathcal{K}_{H}(x,y)\dm(y)\ \in L^2(H).
 	\]
 	 \end{description}
Note that if $\varphi_1,\ldots,\varphi_N$ is an orthonormal basis of $H$, then we can write
\begin{equation}\label{eq:Kerndef}
 \mathcal{K}_{H}(x,y)=\sum_{j=1}^N\varphi_j(x)\overline{\varphi_j(y)},
\end{equation}
 and clearly
 \begin{equation*}
 \int_{\SO}\mathcal{K}_{H}(x,x)\dm(x)=N.
 \end{equation*}
 Coming back to the case of interest and following ideas in \cite{Beltran}, we choose as subspace $H$  the span of the first eigenspaces of $\Delta_g$.
\begin{lemma}\label{lem:KH}
	Let $L\geq0$ and $H_L\subseteq L^2(\SO)$ be the span of the union of eigenspaces for eigenvalues $\lambda_0,\ldots,\lambda_L$ of $\Delta_g$. Then, we define
	\begin{equation*}
	N:=\dim(H_L)=\binom{2L+3}{3}=\mathcal{C}_{2L}^{(2)}(1)=\frac{4}{3}L^3+O(L^2).
	\end{equation*}
	Moreover\footnote{Here, $\mathcal{C}_{2L}^{(2)}$, $L\geq0$, 
		is the sequence of Gegenbauer (ultra-spherical) polynomials.}, the projection kernel is:
	\[
\mathcal{K}_L(\alpha,\beta)=\mathcal{C}_{2L}^{(2)}\Big(
\cos\big(\tfrac{\omega(\alpha^{-1}\beta)}{2}\big)\Big).
\]

\end{lemma}
\subsection{Chebyshev polynomials and proofs of lemmas}\label{subsec_proofs}
The degree $n+1$ Chebyshev polynomials of first and second kind satisfy the recurrence relation 
\begin{equation}\label{eq_ChebyshevRecurence}
	P_{n+1}(x)=2xP_n(x)-P_{n-1}(x),
\end{equation}
with $\mathcal{T}_0\equiv 1$, $\mathcal{T}_1(x)=x $ and $\mathcal{U}_{-1}\equiv 0$, $\mathcal{U}_0(x)\equiv 1 $ in their respective notation. Gegenbauer or ultra-spherical polynomials $\mathcal{C}_{n}^{(\lambda)}(x)$ of degree $n$ and index $\lambda$ appear any time rotation invariance plays a role, and can be defined for integer $\lambda$ as multiples of derivatives of Chebyshev polynomials of the second kind; sufficient for us is the one formula in \eqref{eq_chebyshevT_basics}.
 With this said, using \eqref{eq:Greendef}, \eqref{eq_chebychevU} and \eqref{eq:Kerndef}, we obtain
\begin{equation}\label{eq_simpliefied_K_G}
    \mathcal{K}(\alpha,\beta)=\sum_{\ell=0}^{L}(2\ell+1)\ \mathcal{U}_{2\ell}
 \Big(\cos\big(\tfrac{\omega(\alpha^{-1}\beta)}{2}\big)\Big);
 \end{equation}
 \begin{equation}\label{eq_green_function}
   \mathcal{G}(\alpha, \beta)=  \sum_{\ell=1}^{\infty} \frac{2\ell+1}{\ell(\ell+1)}\mathcal{U}_{2\ell}
 \Big(\cos\big(\tfrac{\omega(\alpha^{-1}\beta)}{2}\big)\Big).
\end{equation}
Further we list some equations for later reference and the reader's 
convenience.
\begin{equation}\label{eq_chebyshevT_basics}
\begin{array}{rll}
  2\mathcal{T}_{2\ell+1}(x)&=\mathcal{U}_{2\ell+1}(x)-\mathcal{U}_{2\ell-1}(x)&   \mbox{ \cite[Eq. 22.5.8]{Stegun}},\\
\mathcal{T}_{n}(1)&=1& \mbox{ \cite[Eq. 8.944.1]{Gradshteyn}},\\
\tfrac{\mathrm{d}}{\mathrm{d}x}\mathcal{T}_{2\ell+1}(x)&=(2\ell+1)\ \mathcal{U}_{2\ell}(x)&\mbox{ \cite[Eq. 8.949.1]{Gradshteyn}},\\
  \tfrac{\mathrm{d}}{\mathrm{d}x}\mathcal{U}_{2L+1}(x)&=2\mathcal{C}_{2L}^{(2)}(x)& \mbox{ \cite[Eq. 8.949.4]{Gradshteyn}},\\
  \mathcal{C}_{n}^{(\lambda)}(1)&=\binom{2\lambda+n-1}{n}& \mbox{ \cite[Eq. 8.937.4]{Gradshteyn}}.\\
\end{array}
\end{equation}
\begin{proof}[Proof of Lemma \ref{lem:KH}]
Let $y:=\cos\big(\tfrac{\omega(\alpha^{-1}\beta)}{2}\big)$, then
by \eqref{eq_simpliefied_K_G} and \eqref{eq_chebyshevT_basics}
\begin{equation*}
  \mathcal{K}(\alpha,\beta)=\frac{\mathrm{d}}{\mathrm{d}x}\sum_{\ell=0}^{L}\mathcal{T}_{2\ell+1}(x)
 \Big|_{y}=\frac{\mathrm{d}}{\mathrm{d}x}\frac{1}{2}\mathcal{U}_{2L+1}(x)
 \Big|_{y}=\mathcal{C}_{2L}^{(2)}\big(y\big).
\end{equation*}
The formula for the dimension of $H_L$ can be proved as follows. The eigenspace associated to $\lambda_\ell=\ell(\ell+1)$ has dimension $(2\ell+1)^2$ since this is the number of elements of its basis $D_{m,n}^\ell$. Thus dim($H_L$) is given by 
$\sum_{\ell=0}^L(2\ell+1)^2$.
\end{proof}

\begin{proof}[Proof of Lemma \ref{lem:Green}] In \eqref{eq_green_function} we apply the
equality 
$$\mathcal{U}_{2\ell}(\cos(t))=\frac{\sin\big((2\ell+1)t\big)}{\sin(t)}\hspace{0.5cm}\mbox{\cite[Eq. 22.3.16]{Stegun}},$$
 and reason, under the assumption $w:=\omega(\alpha^{-1}\beta)\neq 0$, as follows
\begin{equation*}
\begin{split}
  \mathcal{G}(\alpha, \beta)&=\sum_{\ell=1}^{\infty} \frac{2\ell+1}{\ell(\ell+1)}\frac{\sin\big((2\ell+1)\tfrac{w}{2}\big)}
 {\sin\big(\tfrac{w}{2}\big)}\\
 &=\frac{1}{\sin\big(\tfrac{w}{2}\big)}\sum_{\ell=1}^{\infty}\left( \frac{\sin\big((2\ell+1)\tfrac{w}{2}\big)}
 {\ell+1}+\frac{\sin\big((2\ell+1)\tfrac{w}{2}\big)}{\ell}\right)\\
  &=  \frac{1}{i}\Big(-\log\big(1-e^{iw}\big)+\log\big(1-e^{-iw}\big)\Big)
  \cot\big(\tfrac{w}{2}\big)-1;
\end{split}
\end{equation*}
where we used the well known fact, that the power series for $\log(1-x)$ at 1 converges at the boundary of its disc of convergence (except for $x=1$) and equals the logarithm at these values:
\begin{equation*}
\begin{split}
\sum_{\ell=1}^{\infty}\frac{\sin\big((2\ell+1)\tfrac{w}{2}\big)}{\ell+1}&=
\frac{1}{2i}\sum_{\ell=1}^{\infty}\frac{e^{i\tfrac{w}{2}(2\ell+1)}-e^{-i\tfrac{w}{2}(2\ell+1)}}{\ell+1}\\
\hspace{-0.5cm}=&\frac{e^{-i\tfrac{w}{2}}}{2i}\sum_{\ell=1}^{\infty}\frac{e^{iw(\ell+1)}}{\ell+1}-
\frac{e^{i\tfrac{w}{2}}}{2i}\sum_{\ell=1}^{\infty}\frac{e^{-iw(\ell+1)}}{\ell+1}\\
\hspace{-0.5cm}=&\frac{-e^{-i\tfrac{w}{2}}}{2i}\big(\log\big(1-e^{iw}\big)+e^{iw}\big)+
\frac{e^{i\tfrac{w}{2}}}{2i}\big(\log\big(1-e^{-iw}\big)+e^{-iw}\big)\\
\hspace{-0.5cm}=&\frac{-e^{-i\tfrac{w}{2}}}{2i}\log\big(1-e^{iw}\big)+
\frac{e^{i\tfrac{w}{2}}}{2i}\log\big(1-e^{-iw}\big)-\sin(\tfrac{w}{2}),\\
\end{split}
\end{equation*}
and similarly
\begin{equation*}
\sum_{\ell=1}^{\infty}\frac{\sin\big((2\ell+1)\tfrac{w}{2}\big)}{\ell}=
\frac{-e^{i\tfrac{w}{2}}}{2i}\log\big(1-e^{iw}\big)+
\frac{e^{-i\tfrac{w}{2}}}{2i}\log\big(1-e^{-iw}\big).
\end{equation*}
Further, by $ 1-e^{-iw}=2ie^{-i\tfrac{w}{2}}\sin(\tfrac{w}{2})$, we conclude
\begin{equation*}
\begin{split}
\log\big(1-e^{-iw}\big)-\log\big(1-e^{iw}\big)&=
\log\big(2ie^{-i\tfrac{w}{2}}\sin(\tfrac{w}{2})\big)-\log\big(-2ie^{i\tfrac{w}{2}}\sin(\tfrac{w}{2})\big)\\
&=
\log\big(2e^{i\tfrac{-w+\pi}{2}}\sin(\tfrac{w}{2})\big)-\log\big(2e^{i\tfrac{w-\pi}{2}}\sin(\tfrac{w}{2})\big)\\
&=(-w+\pi)\tfrac{i}{2}-(w-\pi)\tfrac{i}{2}=i(\pi-w),
\end{split}
\end{equation*}
where we used a property of the
complex logarithm: $\log(re^{i\varphi})=\log(r)+i\varphi$.
\end{proof}

 \section{Riesz $s$-Energy: Proof of Theorem \ref{thm_Riesz3}}\label{appendix2}
 Recall if $A$ is a real matrix, we have $\|A\|_F^2:=\mbox{\textbf{Trace}}(A^tA)$. We set throughout $N=N(L)=\mathcal{C}^{(2)}_{2L}(1)$ for $L\in\N$, and note next a well known fact before we proceed, see for instance \cite[Eq. (33)]{Huynh}.
 \begin{lemma}\label{lem_omegaFrobenius}
  For $\alpha,\beta\in\SO$, we have
  $\|\alpha-\beta\|_F=\sqrt{8}\sin\big(\tfrac{\omega(\alpha^{-1}\beta)}{2}\big). $
 \end{lemma}
\begin{proof}
We abbreviate $\omega=\omega(\alpha^{-1}\beta)$, and use the half-angle formula for sine:
\begin{equation*}
\begin{split}
 \|\alpha-\beta\|_F^2&=\mbox{\textbf{Trace}}\big[(\alpha-\beta)^t(\alpha-\beta)\big]=6-2\mbox{\textbf{Trace}}(\alpha^{-1}\beta)\\
 &=8\frac{2-\big(\mbox{\textbf{Trace}}(\alpha^{-1}\beta)-1\big)}{4}=
 8\frac{1-\cos(\omega)}{2}
 =
 8\big[\sin\big(\tfrac{\omega}{2}\big)\big]^2.\qedhere
 \end{split}
\end{equation*}
\end{proof}
Reminding us first of the Beta function $\mathcal{B}(a,b):=\int_0^1 t^{a-1}(1-t)^{b-1}\dt$ for $a,b>0$, we are ready to state our first proposition.
\begin{proposition}
   For $s\in(0,3)$ and $N=N(L)$, we have
  $$ \E^s_R(N)\ \leq\  \tfrac{2}{8^{s/2}\pi}\mathcal{B}\big(\tfrac{3-s}{2},\tfrac{1}{2}\big)N^2 +O(N^{1+s/3}).$$
  If $s\in\{1,2\}$, we have more information on the term $O(N^{1+s/3})$: It is respectively
  $$-\tfrac{\sqrt{2}}{\pi}\big(\tfrac{3}{4}\big)^{4/3}N^{4/3}+O(N)\mbox{ \hspace{0.1cm} and \hspace{0.2cm}}-\tfrac{4}{15}\big(\tfrac{3}{4}\big)^{5/3}N^{5/3}+O(N^{4/3}). $$
\end{proposition}
\begin{proof}
We use \eqref{eq_expected_value_dpp}, Lemma \ref{lem:KH}, Lemma  \ref{lem_omegaFrobenius}, invariance of Haar measure, and \eqref{eq_integralSO}:
\begin{equation*}
\begin{split}
  \iint_{\SO}&\frac{\mathcal{K}(\alpha,\alpha)^2-
 \mathcal{K}(\alpha,\beta)^2}{\|\alpha-\beta\|^s_F}\dm(\alpha,\beta)\\
 &=\iint_{\SO}\frac{\big[\mathcal{C}_{2L}^{(2)}(1)\big]^2-
 \big[\mathcal{C}_{2L}^{(2)}\big(\cos\big(\frac{\omega(\alpha^{-1}\beta)}{2}\big)\big)\big]^2}
 {8^{\frac{s}{2}}\big[\sin\big(\tfrac{\omega(\alpha^{-1}\beta)}{2}\big)\big]^s}\dm(\alpha,\beta)\\
&=\frac{2}{8^{\frac{s}{2}}\pi}\int_0^{\pi}\left(N^2-
 \big[\mathcal{C}_{2L}^{(2)}\big(\cos\big(\tfrac{t}{2}\big)\big)\big]^2\right)\sin\big(\tfrac{t}{2}\big)^{2-s}\dt\\
&=\frac{4}{8^{\frac{s}{2}}\pi}N^2\int_0^{\pi/2}\sin(t)^{2-s}\dt-
 \frac{4}{8^{\frac{s}{2}}\pi}\int_0^{1}
  \big[\mathcal{C}_{2L}^{(2)}(t)\big]^2(1-t^2)^{\frac{1-s}{2}}\dt.\\
\end{split}
\end{equation*} 
The next line is, apart of the factor $\tfrac{4}{8^{s/2}\pi} N^2$, the continuous Riesz $s$-energy:
\begin{equation*}
 \begin{split}
  \int_0^{\pi/2}\sin(t)^{2-s}\dt=\int_0^{1}\frac{t^{1-s}\cdot t}{\sqrt{1-t^2}}\dt=\frac{1}{2}\int_0^{1}t^{\frac{1-s}{2}}(1-t)^{-1/2}\dt=\onehalf\mathcal{B}\big(\tfrac{3-s}{2},\tfrac{1}{2}\big).
 \end{split}
\end{equation*}
For $0<s<3$, 
we hence obtain
  \begin{equation*}
   \begin{split}
    \int_0^{1}
  \big[\mathcal{C}_{2L}^{(2)}(t)\big]^2&(1-t^2)^{\frac{1-s}{2}}\dt=\int_0^{\pi/2}\big[C_{2L}^{(2)}(\cos(t))\big]^2\sin(t)^{2-s}\dt\\
  &\leq \int_0^{1/L}\big[C_{2L}^{(2)}(\cos(t))\big]^2t^{2-s}\dt+\int_{1/L}^{\pi/2}\big[C_{2L}^{(2)}(\cos(t))\big]^2t^{2-s}\dt\\
   &\leq \big[C_{2L}^{(2)}(1)\big]^2\frac{t^{3-s}}{3-s}\Big|_0^{1/L}-\frac{CL^2}{1+s}\frac{1}{t^{1+s}}\Big|_{1/L}^{\pi/2}=O(L^{3+s});
   \end{split}
  \end{equation*}
  where we inferred \cite[ Eq. 7.33.6]{Szego}, i.e. for every $c>0$ there is $C\geq0$ such that 
  \begin{equation*}
  |C_{2L}^{(2)}(\cos( \theta))|\leq \frac{CL}{\theta^2},\quad \frac{c}{L}\leq\theta\leq\frac{\pi}{2}.
  \end{equation*}
The case  $s=1$ is Lemma \ref{lem_gegenbauerAsymptotics};  the case $s=2$ follows from Lemma \ref{lem_GegenbauerCosineSquared}:
  \begin{equation*}
   \begin{split}
     \int_0^{1}\frac{
  \big[\mathcal{C}_{2L}^{(2)}(t)\big]^2}{\sqrt{1-t^2}}\dt&= \int_0^{\pi/2}
  \big[\mathcal{C}_{2L}^{(2)}(\cos(t))\big]^2\dt= \sum_{u=0}^{2L}c_{u,u}\frac{\pi}{2}=\frac{8\pi }{15}L^5 +O(L^4),
   \end{split}
  \end{equation*}
  where $c_{u,u}=c^2_{u,u}(2L)$ with notation as in Lemma \ref{lem_GegenbauerCosineSquared}.
  \end{proof}
 To use \eqref{eq_integralSO} in the next proof, which is valid for $L^1$ functions  -- we argue as follows: Use Lebesgue's monotone convergence theorem with \eqref{eq_integralSO} on $f_n=\min\{n,f\}\ra f$. We will further use the digamma function $\psi$, see Appendix \ref{appendix}.
\begin{proof}[Proof of Theorem \ref{thm_Riesz3}] We proceed as in the previous proof and use Lemma \ref{lem_GegenbauerCosineSquared}:
  \begin{equation*}
   \begin{split}
\int_0^{\pi/2}\frac{\big[\mathcal{C}_{2L}^{(2)}(1)\big]^2-
 \big[\mathcal{C}_{2L}^{(2)}\big(\cos(t)\big)\big]^2}{\sin(t)}&\dt=
 2\sum_{r=1}^{2L}\int_0^{\pi/2}\frac{1-\cos(2rt)}{\sin(t)}\dt\sum_{u=0}^{2L-r}c_{r+u,u}\\
&\hspace{-1cm}=4\sum_{r=1}^{2L}\int_0^{\pi/2}\mathcal{U}_{r-1}\big(\cos(t)\big)^2\sin(t)\dt\sum_{u=0}^{2L-r}c_{r+u,u}\\
&\hspace{-1cm}=4\sum_{r=1}^{2L}\int_0^{1}\mathcal{U}_{r-1}(t)^2\dt\sum_{u=0}^{2L-r}c_{r+u,u}=(\star).
    \end{split}
  \end{equation*}
  We use \eqref{eq_digamma_sum}:
  $  \int_0^{1}\mathcal{U}_{n}(t)^2\dt=\onehalf\Big(\psi(n+\tfrac{3}{2})+\gamma+\log(4)\Big),
  $
  and obtain
  \begin{equation*}
  \begin{split}
  (\star)&=2(\gamma+\log(4))\sum_{r=1}^{2L}\sum_{u=0}^{2L-r}c_{r+u,u}+
  2\sum_{r=1}^{2L}\psi\big(r+\tfrac{1}{2}\big)\sum_{u=0}^{2L-r}c_{r+u,u}=:S_1+S_2.
  \end{split}
  \end{equation*}
By $c_{r+u,u}=c^2_{r+u,u}(2L)=(r+u+1)(2L-r-u+1)(u+1)(2L-u+1)$, we have
 $$
  \sum_{u=0}^{2L-r}c_{r+u,u}=\frac{16}{15}L^5+\frac{2}{3}L^2r^3-\frac{4}{3}L^3r^2-\frac{r^5}{30}+O_{a+b<5}(L^ar^b),
$$
 and hence by well known summation formulas due to Faulhaber:
\begin{equation*}
 \begin{split}
  S_1&=2\big(\gamma+\log(4)\big)\left(\frac{16}{15}L^5 2L+\frac{2}{3}L^24L^4-\frac{4}{3}L^3 \frac{8}{3}L^3-\frac{1}{30}\frac{32}{3}L^6\right)+O(L^5)\\
  &=\frac{16}{9}\big(\gamma+\log(4)\big) L^6+O(L^5).
 \end{split}
\end{equation*}
Invoking Lemma \ref{lem_DigammaAsymptotic} yields
\begin{equation*}
\begin{split}
\onehalf S_2&=
\frac{16}{15}L^5\cdot \big(2L\ \psi(2L)  -2L\big)  
+\frac{2}{3}L^2\cdot\Big(\frac{(2L)^4}{4}\psi(2L)  -\frac{(2L)^{4}}{4^2}\Big)\\
&\hspace{0.5cm}  -  \frac{4}{3}L^3\cdot \Big(\frac{(2L)^{3}}{3}\psi(2L)  -\frac{(2L)^{3}}{3^2}\Big)
-\frac{1}{30}\Big(\frac{(2L)^{6}}{6}\psi(2L)  -\frac{(2L)^{6}}{6^2}\Big)\\
&\hspace{0.5cm}+O(L^5\log(L))\\
&=\frac{8}{9}L^6\cdot\psi(2L)-\frac{14}{9}L^6+O(L^5\log(L)).
\end{split}
\end{equation*}
Since $N^2=\mathcal{C}_{2L}^{(2)}(1)^2=\frac{16}{9}L^6\big(1+O(L^{-1})\big)$, and $\big(\frac{3}{4}N\big)^{1/3}=L\big(1+O(L^{-1})\big)^{1/6}$ we see
$$ \frac{1}{3}\log\big(\tfrac{3}{4}N\big)=\log(L)+O(L^{-1});$$
and with \eqref{eq_digamma_sum}, using harmonic numbers $H_n:=\sum_{k=1}^n \frac{1}{k}= \log(n)+\gamma+O(n^{-1})$:
\begin{equation*}
 \begin{split}
  (\star)&= \frac{16}{9}L^6\cdot\Big(\psi(2L)+\gamma+\log(4)\Big)-\frac{7}{4}\frac{16}{9}L^6+O(L^5\log(L))\\
  &= 2N^2\cdot\Big(H_{4L+1}-\onehalf H_{2L}\Big)-\frac{7}{4}N^2+O(N^{5/3}\log(N))\\
  &= 2N^2\cdot\log\Big(\frac{4L+1}{\sqrt{2L}}\Big)+\frac{4\gamma-7}{4}N^2+O(N^{5/3}\log(N))\\
   &= N^2\cdot\log(L)+\frac{4\gamma+4\log(8)-7}{4}N^2+O(N^{5/3}\log(N))\\
   &= \frac{1}{3}N^2\cdot\log(N)+\frac{1}{3}
   \Big(3\gamma+\log\big(8^3\tfrac{3}{4}\big)-\frac{21}{4}\Big)N^2+O(N^{5/3}\log(N));\\
 \end{split}
\end{equation*}
proving the claim when multiplied by $\tfrac{4}{8^{3/2}\pi}$.
\end{proof}

\begin{lemma}\label{lem_DigammaAsymptotic}
 Let $\psi(t)$ be the digamma function and $m\geq 0$, then
 $$\sum_{k=1}^n k^m \psi\big(k+\tfrac{1}{2}\big)=\frac{n^{m+1}}{m+1}\psi(n) -\frac{n^{m+1}}{(m+1)^2}+O(n^m\log(n)).$$
\end{lemma}
\begin{proof}
Since $\psi(t)=\log(t)+O(\tfrac{1}{t})$ for $t>2$, we have
$$\sum_{k=1}^n k^m \psi\big(k+\tfrac{1}{2}\big)
=\int_1^n t^m\log(t)\dt+O(n^{m}\log(n));$$
as the sum can be bounded from above and below by the same integral, apart from integration boundaries, where we obtain the error term. We finish by applying the anti-derivative:
$\frac{t^{m+1}}{m+1}\log(t)- \frac{t^{m+1}}{(m+1)^2}$.
\end{proof}

\section{Green Energy: Proof of Theorem \ref{thm_Green}}
We prove the lower and upper bound separately in the following two sections.

\subsection{Estimate of the Green Energy: Lower Bound}
We follow an exposition due to N. Elkies, found in \cite[pp. 149-154]{Lang}. This has been pointed out to the authors by E. Saff, and his help is thankfully acknowledged.

The idea is to find a function with nice properties smaller than $\G$, and to bound its energy from below. For $\alpha,\beta\in\SO$ and $t>0$, the following will do:
\begin{equation*}
\mathcal{G}_t(\alpha,\beta)=\sum_{\ell=1}^{\infty}e^{-\ell(\ell+1)\cdot t}\frac{2\ell+1}{\ell(\ell+1)}\mathcal{U}_{2\ell}
\Big(\cos\big(\tfrac{\omega(\alpha^{-1}\beta)}{2}\big)\Big).
\end{equation*}
To show that it really is smaller, we infer an adaptation of \cite[Lem. 5.2]{Lang}.
\begin{lemma}[N. Elkies]\label{lem_Elkies}
	For all $t>0$ and $\alpha\neq \beta$ we have
	$$\mathcal{G}(\alpha,\beta)\geq\mathcal{G}_t(\alpha,\beta)-t.$$
\end{lemma}
\begin{proof}
		Using uniform convergence, we differentiate term by term and define
	$$ h_t(\alpha,\beta):=-\partial_t\G_t(\alpha,\beta)=\sum_{\ell=1}^{\infty}e^{-\ell(\ell+1)\cdot t}(2\ell+1)\sum_{m=-\ell}^\ell\sum_{n=-\ell}^\ell \mathcal{D}_{m,n}^\ell(\alpha)
	\overline{\mathcal{D}_{m,n}^\ell(\beta)}.$$
	Given a smooth test function $\phi$, with uniformly converging representation as $\sum^{\infty}_{\ell=0}\phi_\ell$, where $\phi_\ell=\sum_{m,n}\varphi_{m,n}^\ell\mathcal{D}_{m,n}^\ell\sqrt{2\ell+1}$, we  set
	$$ u(\alpha,t):=\int_{\SO}h_t(\alpha,\beta)\phi(\beta)\dm(\beta)=\sum_{\ell=1}^{\infty}e^{-\ell(\ell+1)\cdot t}\phi_\ell(\alpha),$$
	where we interchanged integration and summation by uniform convergence and used that $\{\mathcal{D}_{m,n}^\ell\sqrt{2\ell+1}\}$ is an orthonormal basis. Now we have uniformly
	$$\lim_{t\ra 0}u(\alpha,t)=\phi(\alpha)-\int_{\SO}\phi(\beta)\dm(\beta)=\phi(\alpha)-\phi_0.$$
	For $t>0$ fixed, we can interchange differentiation and integration yielding
	$$\Delta_gu(\alpha,t)+\partial_tu(\alpha,t)=0.$$
	By the strong maximum principle Theorem \ref{thm_manifoldstrongmaximum}, we have for every $t>0$:
		\begin{equation*}
		\min_{\alpha\in\SO}u(\alpha,t)\geq\min_{\alpha\in\SO}u(\alpha,0).
		\end{equation*}
    The same PDE and estimates hold for
	$$ v(\alpha,t)=u(\alpha,t)+\phi_0.$$
	If $\phi\geq0$, then so is $v(\alpha,t)$ for all $t>0$ by the maximum principle as $v(\alpha,0)=\phi(\alpha)$. Hence 
	$$u(\alpha,t)=v(\alpha,t)-\phi_0\geq-\phi_0 \hspace{0.3cm}\mbox{ for }\phi\geq0.$$
	We further set
	$$\mathds{I}(\alpha,t):=\int_{\SO}\G_t(\alpha,\beta)\phi(\beta)\dm(\beta)=\sum_{\ell=1}^{\infty}e^{-\ell(\ell+1)\cdot t}\frac{\phi_\ell(\alpha)}{\ell(\ell+1)},$$
	where we interchanged sum and integral again. The limit $t\ra 0$ exists and equals the integral of $\G(\alpha,\beta)\cdot\phi(\beta)$. Differentiating term-wise for $t>0$ yields
	$$ \partial_t\mathds{I}(\alpha,t)=-\sum_{\ell=1}^{\infty}e^{-\ell(\ell+1)\cdot t}\phi_\ell(\alpha)=-u(\alpha,t)\leq\phi_0\hspace{0.3cm}\mbox{ for }\phi\geq0.$$
	Finally, for fixed $\alpha$ let $t>\epsilon>0$, then by the fundamental theorem of calculus:
	\begin{equation*}
		\lim_{\epsilon\ra 0}\mathds{I}(\alpha,t)-\mathds{I}(\alpha,\epsilon)=\lim_{\epsilon\ra 0}\int_{\epsilon}^{t}-u(\alpha,t)\dt\leq \phi_0\cdot t
	\end{equation*}
	and thus, for all non-negative test functions $\phi$
	\begin{equation*}
	\int_{\SO}\Big(\G_t(\alpha,\beta)-\G(\alpha,\beta)-t\cdot 1\Big)\phi(\beta)\dm(\beta)\leq 0.
	\end{equation*}
	Since $\G(\alpha,\beta)$ is continuous and locally integrable in $\beta$ away of $\alpha$, this proves the lemma.
\end{proof}
Now by Lemma \ref{lem_Elkies}, we have for some $t>0$ which will be determined later, and some collection of distinct points $\{\alpha_1,\ldots,\alpha_N\}\subset\SO$:
\begin{equation*}
\begin{split}
& \sum_{s\neq k}^{N}\mathcal{G}(\alpha_s,\alpha_k)+N(N-1)2t\geq
\sum_{s\neq k}^{N}\mathcal{G}_{2t}(\alpha_s,\alpha_k)\\
&=\sum_{\ell=1}^{\infty}\frac{2\ell+1}{\ell(\ell+1)}
\sum_{m=-\ell}^\ell\sum_{n=-\ell}^\ell\sum_{s\neq k}^{N}e^{-\ell(\ell+1)\cdot 2t}
\mathcal{D}_{m,n}^\ell(\alpha_s)
\overline{\mathcal{D}_{m,n}^\ell(\alpha_k)}\ \ =\\
&\sum_{\ell=1}^{\infty}\frac{2\ell+1}{\ell^2+\ell}\sum_{m=-\ell}^\ell\sum_{n=-\ell}^\ell 
\left(\bigg|\sum_{k=1}^{N}e^{-\ell(\ell+1)\cdot t}
\mathcal{D}_{m,n}^\ell(\alpha_k)\bigg|^2-\sum_{k=1}^{N}e^{-\ell(\ell+1)\cdot 2t}
\Big|\mathcal{D}_{m,n}^\ell(\alpha_k)\Big|^2\right)\\
&\geq -\sum_{\ell=1}^{\infty}\frac{2\ell+1}{\ell(\ell+1)}\sum_{m=-\ell}^\ell\sum_{n=-\ell}^\ell 
\sum_{k=1}^{N}e^{-\ell(\ell+1)\cdot 2t}
\Big|\mathcal{D}_{m,n}^\ell(\alpha_k)\Big|^2=-N\mathcal{G}_{2t}(\alpha,\alpha).
\end{split}
\end{equation*}
Thus our remaining task is to find an asymptotic for $\G_t(\alpha,\alpha)$ in $t$.
First we note that
$$\frac{e^{-\ell(\ell+1)\cdot t}}{\ell(\ell+1)}=4\frac{e^{-\ell(\ell+1)\cdot t}}{(2\ell+1)^2}\Big(1+\frac{1}{4\ell(\ell+1)}\Big)
=4\frac{e^{-\ell(\ell+1)\cdot t}}{(2\ell+1)^2}+\frac{C_\ell}{\ell^4},$$
where $C_\ell<1/4$ is  some constant. For $0<t\ll1$ we then obtain
\begin{equation}\label{eq_GreenTea}
\begin{split}
\mathcal{G}_t(\alpha,\alpha)&=\sum_{\ell=1}^{\infty}e^{-\ell(\ell+1)\cdot t}\frac{(2\ell+1)^2}{\ell(\ell+1)}=\sum_{\ell=1}^{\infty}\left(e^{-\ell(\ell+1)\cdot t}4+\frac{e^{-\ell(\ell+1)\cdot t}}{\ell(\ell+1)}\right)\\
&=4e^{t/4}\int_{0}^{\infty}e^{-(2x+1)^2 t/4}+\frac{e^{-(2x+1)^2 t/4}}{(2x+1)^2}\dx+O(1)\\
 &=2e^{t/4}\int_{1}^{\infty}e^{-x^2 t/4}+\frac{t}{4}\frac{e^{-x^2 t/4}}{x^2t/4}\dx+O(1)\\
 &=\frac{4e^{t/4}}{\sqrt{t}}\int_{\sqrt{t}/2}^{\infty}e^{-x^2}+\frac{t}{4}\frac{e^{-x^2}}{x^2}\dx+O(1)\\
 &= 2e^{t/4}\sqrt{\frac{\pi}{t}}-
 \sqrt{t}e^{t/4}\Big[\sqrt{\pi}\cdot\mbox{erf}(x)+\frac{e^{-x^2}}{x}\Big]_{\sqrt{t}/2}^{\infty}+O(1)\\
 &=2\sqrt{\frac{\pi}{t}}+O(1);
\end{split}
\end{equation}
with 
$$\mathrm{erf}(x):=\frac{2}{\sqrt{\pi}}\int_{0}^{x}e^{-y^2}\dy. $$

If we choose $t=\tfrac{\sqrt[3]{\pi}}{2N^{2/3}}$, then by \eqref{eq_GreenTea}
\begin{equation*}
 \mathcal{G}_{2t}(\alpha,\alpha)=2\sqrt[3]{\pi} N^{\frac{1}{3}}+O(1),
\end{equation*}
and hence
$$\sum_{s\neq k}^{N}\mathcal{G}(\alpha_s,\alpha_k)\geq-3\sqrt[3]{\pi}N^{\frac{4}{3}}+O(N),$$
proving the lower bound in Theorem \ref{thm_Green}.
\subsection{Estimate of the Green Energy: Upper Bound}
According to \eqref{eq_expected_value_dpp}, we have to estimate the integral
\begin{equation*}
 I=\iint_{\SO}\mathcal{G}(\alpha,\beta)\left(\mathcal{K}(\alpha,\alpha)^2-
 \mathcal{K}(\alpha,\beta)^2\right)\dm(\alpha,\beta),
\end{equation*}
which by Lemmas \ref{lem:Green} and \ref{lem:KH} and by invariance of  Haar measure equals 
\begin{equation*}
 \int_{\SO}\Big(\big(\pi-\omega(\alpha)\big)
  \cot\big(\tfrac{\omega(\alpha)}{2}\big)-1\Big)\left(\mathcal{C}_{2L}^{(2)}(1)^2-
 \Big[\mathcal{C}_{2L}^{(2)}\big(\cos\big(\tfrac{\omega(\alpha)}{2}\big)\big)\Big]^2\right)\dm(\alpha).
\end{equation*}
The integrand is in $L^1(\SO)$ since the singularity of the cotangent is removed by the zero of the difference of Gegenbauer polynomials, thus being a continuous function on a compact set. We hence can apply \eqref{eq_integralSO} getting:
  \begin{equation*}
I= \frac{2}{\pi}\int_0^{\pi}\Big(\big(\pi-t\big)
  \cot\big(\tfrac{t}{2}\big)-1\Big)\left(\mathcal{C}_{2L}^{(2)}(1)^2-
 \Big[\mathcal{C}_{2L}^{(2)}\big(\cos\big(\tfrac{t}{2}\big)\big)\Big]^2\right)\sin^2\big(\tfrac{t}{2}\big)\dt.
\end{equation*}
Since 
$$\int_0^{\pi}\Big(\big(\pi-t\big)
  \cot\big(\tfrac{t}{2}\big)-1\Big)\sin^2\big(\tfrac{t}{2}\big)\dt=0,$$
  we indeed have
    \begin{equation}\label{eq_energy_integral}
    -I=\frac{2}{\pi}\int_0^{\pi}\Big(\big(\pi-t\big)
  \cot\big(\tfrac{t}{2}\big)-1\Big)\Big[\mathcal{C}_{2L}^{(2)}\big(\cos\big(\tfrac{t}{2}\big)\big)\Big]^2\sin^2\big(\tfrac{t}{2}\big)\dt.
\end{equation}   
We simplify by noticing that
\begin{equation*}
\begin{split}
  \int_0^{\pi}\Big[\mathcal{C}_{2L}^{(2)}\big(\cos\big(\tfrac{t}{2}\big)\big)\Big]^2
  \sin^2\big(\tfrac{t}{2}\big)\dt&=2\int_{0}^{1}
  \big[\mathcal{C}_{2L}^{(2)}(t)\big]^2 \sqrt{1-t^2}\dt\\
  &=\int_{-1}^{1}
  \big[\mathcal{C}_{2L}^{(2)}(t)\big]^2 \sqrt{1-t^2}\dt\\
  &= \int_{-1}^{1}  \big[\mathcal{C}_{2L}^{(2)}(t)\big]^2  \sqrt{1-t^2}(1+t)\dt,
\end{split}
\end{equation*}
where we used that odd functions integrate to zero over symmetric
intervals.
But
\begin{equation}\label{eq_GegenbauerSpecialFormula}
 \int_{-1}^{1}
   \big[\mathcal{C}_{2L}^{(2)}(t)\big]^2  \sqrt{1-t}(1+t)^{3/2}\dt=\frac{\pi}{2}\binom{2L+3}{2L},
\end{equation}
by the following equality, valid for $\nu>\tfrac{1}{2}$ and
 found in \cite[Eq. 7.314, p.789]{Gradshteyn}:
 \begin{equation}\label{eq_equation_Gegenbauer}
  \int_{-1}^1 (1-x)^{\nu-\tfrac{3}{2}}(1+x)^{\nu-\tfrac{1}{2}}
  \big|\mathcal{C}_{n}^{(\nu)}(x)\big|^2\dx=\frac{\pi^{1/2}
  \Gamma(\nu-\tfrac{1}{2})\Gamma(2\nu+n)}{n!\Gamma(\nu)\Gamma(2\nu)}.
 \end{equation}
 We have then proved that
 \begin{equation*}
  \begin{split}
   -I&=\frac{2}{\pi}\int_0^{\pi}\big(\pi-t\big)
  \cot\big(\tfrac{t}{2}\big)\Big[\mathcal{C}_{2L}^{(2)}\big(\cos\big(\tfrac{t}{2}\big)\big)\Big]^2
  \sin^2\big(\tfrac{t}{2}\big)\dt+O(L^3)\\
  &=\frac{4}{\pi}\int_{0}^{1}\big(\pi-2\cos^{-1}(t)\big)
  \cdot t\cdot\big[\mathcal{C}_{2L}^{(2)}(t)\big]^2\dt+O(L^3)\\
  &=4\int_{0}^{1}t\cdot\big[\mathcal{C}_{2L}^{(2)}(t)\big]^2\dt
  -\frac{4}{\pi}\int_{0}^{1}2\cos^{-1}(t)
  \cdot t\cdot\big[\mathcal{C}_{2L}^{(2)}(t)\big]^2\dt+O(L^3).
  \end{split}
 \end{equation*}
 Next we use Lemma \ref{lem_gegenbauerIdentity} and Lemma \ref{lem_gegenbauerAsymptotics} in
 $$ \int_{0}^{1}t^2\cdot\big[\mathcal{C}_{2L}^{(2)}(t)\big]^2\dt<
 \int_{0}^{1}t\cdot\big[\mathcal{C}_{2L}^{(2)}(t)\big]^2\dt<\int_{0}^{1}\big[\mathcal{C}_{2L}^{(2)}(t)\big]^2\dt,$$
 and obtain
 $$  \int_{0}^{1}t\cdot\big[\mathcal{C}_{2L}^{(2)}(t)\big]^2\dt=L^4+O(L^3).$$
Finally we use
 \begin{equation*}
 0\leq 2\cos^{-1}(t)\leq \pi\sqrt{1-t}, \hspace{0.5cm} \mbox{ for }t\in[0,1]
\end{equation*}
so that, by \eqref{eq_GegenbauerSpecialFormula}
\begin{equation*}
 \begin{split}
\int_{0}^{1}2\cos^{-1}(t)
  \cdot t\cdot\big[\mathcal{C}_{2L}^{(2)}(t)\big]^2\dt&<
  \int_{0}^{1}\pi\sqrt{1-t}
  \cdot t\cdot\big[\mathcal{C}_{2L}^{(2)}(t)\big]^2\dt\\
  &<\pi\int_{-1}^{1}
   \big[\mathcal{C}_{2L}^{(2)}(t)\big]^2  \sqrt{1-t}(1+t)^{3/2}\dt=O(L^3).
 \end{split}
\end{equation*}
Hence
$$ I=-4L^4+O(L^3),$$
and the upper bound in Theorem \ref{thm_Green}  follows from $N=\frac{4}{3}L^3+O(L^2)$.


\section{Variance: Proof of Theorem \ref{thm_variance}}
Let $A=B(\fatone,2\varepsilon)\subseteq\SO$ be  as in the introduction, namely
\[
A=\{\beta\in\SO:\omega(\beta)<2\varepsilon\}=\big\{\beta\in\SO\ : \|\beta-\fatone\|_F<\sqrt{8}\sin(\varepsilon) \big\},
\]
where equality follows from Lemma \ref{lem_omegaFrobenius}. Note that by rotation invariance it suffices to study the variance of the random variable
\[
\eta_{A}=\sum_{k=1}^N\chi_{A}(\alpha_k),
\]
where $\alpha_1,\ldots,\alpha_N$ are generated by our dpp. The expected value of $\eta_\epsilon$ satisfies $\mathbb{E}[\eta_\epsilon]=\mu(A)N$, and the variance of $\eta_{A}$ is by definition (using $\chi_{A}(\alpha_k)^2=\chi_{A}(\alpha_k)$):
\[
\var(\eta_{A})=\mathbb{E}[\eta_{A}^2]-\mathbb{E}[\eta_{A}]^2=\mathbb{E}\bigg[\sum_{i\neq j}\chi_{A}(\alpha_i)\chi_{A}(\alpha_j)\bigg]+\mu(A)N-\mu(A)^2N^2.
\]
The expected value of the right-hand side equals by \eqref{eq_expected_value_dpp}, with $f(x,y)=\chi_A(x)\chi_A(y)$
\begin{multline*}
\iint_{\alpha,\beta\in A}\big[\mathcal{C}_{2L}^{(2)}(1)\big]^2-
\big[\mathcal{C}_{2L}^{(2)}\big(\cos\big(\tfrac{\omega(\alpha^{-1}\beta)}{2}\big)\big)\big]^2\dm(\beta,\alpha)=\\\mu(A)^2N^2-\iint_{\alpha,\beta\in A}
\big[\mathcal{C}_{2L}^{(2)}\big(\cos\big(\tfrac{\omega(\alpha^{-1}\beta)}{2}\big)\big)\big]^2\dm(\beta,\alpha).
\end{multline*}
In other words, we have
\[
\var(\eta_{A})=\mu(A)N-\iint_{\alpha,\beta\in A}
\big[\mathcal{C}_{2L}^{(2)}\big(\cos\big(\tfrac{\omega(\alpha^{-1}\beta)}{2}\big)\big)\big]^2\dm(\beta,\alpha),
\]
and therefore, using invariance of Haar measure, \eqref{eq_integralSO} and \eqref{eq_GegenbauerSpecialFormula}
\begin{equation*}
	\begin{split}
	\var(\eta_{A})-&\int_A\int_{A^c} \big[\mathcal{C}_{2L}^{(2)}\big(\cos\big(\tfrac{\omega(\alpha^{-1}\beta)}{2}\big)\big)\big]^2\dm(\beta)\dm(\alpha)\\
	&=\mu(A)N-\int_{\SO}\chi_A(\alpha)\int_{\SO} \big[\mathcal{C}_{2L}^{(2)}\big(\cos\big(\tfrac{\omega(\beta)}{2}\big)\big)\big]^2\dm(\beta)\dm(\alpha)\\
	&=
	\mu(A)N-\int_{\SO}\chi_A(\alpha)\cdot N\dm(\alpha)=0.
	\end{split}
\end{equation*}
All in one we have proved the variance version of \cite[Eq. 28]{Rider}:
\[
\var(\eta_{A})=\int_A\int_{A^c} \big[\mathcal{C}_{2L}^{(2)}\big(\cos\big(\tfrac{\omega(\alpha^{-1}\beta)}{2}\big)\big)\big]^2\dm(\beta)\dm(\alpha).
\]
Now, note that
$$A^c=\big\{\beta\in\SO\ : \|\beta-\fatone\|_F\geq\sqrt{8}\sin(\varepsilon) \big\},$$
and by the triangle inequality: 
$ \|\beta-\fatone\|_F \leq \|\beta-\alpha\|_F+\|\fatone-\alpha\|_F$ for $\alpha\in A$, we see
\begin{equation*}
  A^c\subset S_{\alpha}:=\big\{ \beta\in\SO\ : \omega(\alpha^{-1}\beta)\geq f\big(\omega(\alpha)\big)\big\},
\end{equation*}
where $f\big(\omega(\alpha)\big):=2\arcsin\big(\sin(\varepsilon)-\sin\big(\tfrac{\omega(\alpha)}{2}\big)\big)$.
 Thus, for the characteristic function $\chi_\alpha$ of $S_{\alpha}$, we  integrate over $\SO$  and use \eqref{eq_integralSO}:
 \begin{equation*}
 \begin{split}
  \int\chi_\alpha(\beta)\Big[\mathcal{C}_{2L}^{(2)}\big(\cos\big(\tfrac{\omega(\alpha^{-1}\beta)}{2}\big)\big)\Big]^2\dm(\beta) &=\int\chi_\alpha(\alpha\beta)\Big[\mathcal{C}_{2L}^{(2)}\big(\cos\big(\tfrac{\omega(\beta)}{2}\big)\big)\Big]^2\dm(\beta)\\  
  &=\frac{4}{\pi}\int_{\frac{f(\omega(\alpha))}{2}}^{\pi/2}\Big[\mathcal{C}_{2L}^{(2)}\big(\cos(t)\big)\Big]^2\sin^2(t)\dt\\
   &=\frac{4}{\pi}\int_0^{\cos(\frac{f(\omega(\alpha))}{2})}\Big[\mathcal{C}_{2L}^{(2)}(t)\Big]^2\sqrt{1-t^2}\dt.
 \end{split}
\end{equation*}
Applying \eqref{eq_integralSO} one more time yields
  \begin{equation*}
 \begin{split}
  \var(\eta_A)&\leq\int_{\SO}\chi_A(\alpha)\int_{\SO}\chi_\alpha(\beta)\mathcal{C}_{2L}^{(2)}\big(\cos\big(\tfrac{\omega(\alpha^{-1}\beta)}{2}\big)\big)^2\dm(\beta)\dm(\alpha)\\
  & =\frac{4}{\pi}\int_{\SO}\chi_A(\alpha)\int_0^{\cos(\frac{f(\omega(\alpha))}{2})}\Big[\mathcal{C}_{2L}^{(2)}(t)\Big]^2\sqrt{1-t^2}\dt\dm(\alpha)\\
  & =\frac{16}{\pi^2}\int_0^{\varepsilon}\sin(x)^2\int_0^{\sqrt{1-(\sin(\varepsilon)-\sin(x))^2}}\Big[\mathcal{C}_{2L}^{(2)}(t)\Big]^2\sqrt{1-t^2}\dt\dx\\
    &=\frac{16}{\pi^2}\int_0^{\varepsilon}\sin(x)^2\int_{\cos(\varepsilon)}^{\sqrt{1-(\sin(\varepsilon)-\sin(x))^2}}\Big[\mathcal{C}_{2L}^{(2)}(t)\Big]^2\sqrt{1-t^2}\dt\dx\\
    &\hspace{1cm}+\frac{16}{\pi^2}\int_0^{\varepsilon}\sin(x)^2\int_0^{\cos(\varepsilon)}\Big[\mathcal{C}_{2L}^{(2)}(t)\Big]^2\sqrt{1-t^2}\dt\dx=:I_1+I_2.\\
  \end{split}
\end{equation*}
Next we change the order of integration, thus for  $t\in[\cos(\varepsilon),1]$, we integrate over $\{t\}\times [z(t),\varepsilon]$, where $z(t):=\arcsin\big(\sin(\varepsilon)-\sqrt{1-t^2}\big)$. We do this since $x\in[z(t),\varepsilon]$ implies $\sqrt{1-(\sin(\varepsilon)-\sin(x))^2}\in[t,1]$. Thus
$$
  I_1=\frac{16}{\pi^2}\int_{\cos(\varepsilon)}^{1}\Big[\mathcal{C}_{2L}^{(2)}(t)\Big]^2\sqrt{1-t^2}\int_{z(t)}^{\varepsilon}\sin(x)^2\dx\dt.
 $$
Further, by a standard estimate and the mean value theorem, we get
\begin{equation*}
 \begin{split}
  \int_{z(t)}^{\varepsilon}\sin(x)^2\dx &\leq \sin(\varepsilon)^2\Big(\arcsin\big(\sin(\varepsilon)\big)-\arcsin\big(\sin(\varepsilon)-\sqrt{1-t^2}\big)\Big)\\
  &\leq\sin(\varepsilon)^2 \frac{\sqrt{1-t^2}}{\cos( \varepsilon)},
 \end{split}
\end{equation*}
and hence by Lemma \ref{lem_gegenbauerIdentity}
\begin{equation*}
 I_1\leq \frac{16\sin(\varepsilon)^2}{\pi^2\cos(\varepsilon)}\int_{0}^{1}\Big[\mathcal{C}_{2L}^{(2)}(t)\Big]^2(1-t^2)\dt= \frac{\sin(\varepsilon)^2}{\cos(\varepsilon)}O(L^2\log(L)).
\end{equation*}
 Using: $\sin(\varepsilon)=\sqrt{1-\cos(\varepsilon)^2}\leq\sqrt{1-t^2}$, Lemma \ref{lem_gegenbauerIdentity}, and $\frac{\sin(x)}{\sin(\varepsilon)}\leq1$ yields
\begin{equation*}
I_2\leq\frac{16}{\pi^2}\int_0^{\varepsilon}\sin(x)^2\int_0^{\cos(\varepsilon)}\Big[\mathcal{C}_{2L}^{(2)}(t)\Big]^2\sqrt{1-t^2}\frac{\sqrt{1-t^2}}{\sin(\varepsilon)}\dt\dx=\varepsilon^2O(L^2\log(L)).
\end{equation*}
Theorem \ref{thm_variance} is now proved.
 \appendix
 	\section{The Strong Maximum Principle on Manifolds}
 We state the classical strong maximum principle Theorem \ref{thm_classicstrongmaximum} for open, bounded, and connected subsets $U\subset\R^n$, and regard second order parabolic partial differential operators $\mathrm{L}+\frac{\partial}{\partial t}$ acting on functions $C^2_1(U\times(0,T])$, i.e. twice differentiable with respect to spatial variables and once w.r.t. time. $T>0$. A special case of this is extended in Theorem \ref{thm_manifoldstrongmaximum}. We set for smooth coefficients:
 \begin{equation}\label{eq_parabolicdiffop}
 \mathrm{L} u(x,t)=-\sideset{}{^n_{i,j}}\sum a_{ij}(x,t)\tfrac{\partial}{\partial x_i}\tfrac{\partial}{\partial x_j}u(x,t)
 +\sideset{}{^n_{j}}\sum b_j(x,t)\tfrac{\partial}{\partial x_j}u(x,t),
 \end{equation}
 and without loss of generality, $a_{ij}(x,t)=a_{ji}(x,t)$.
 \begin{definition}
 	 $\mathrm{L}+\frac{\partial}{\partial t}$ is said to be uniformly parabolic if there is a $C>0$, s.t.
 	\begin{equation}\label{eq_uniformelliptic}
 	\sum_{i,j} a_{ij}(x,t)\xi_i\xi_j\geq C\|\xi\|_2^2, \hspace{0.3cm}\mbox{where }\xi\in\R^n,\ (x,t)\in U\times(0,T].
 	\end{equation}
 \end{definition}

 \begin{theorem}[Thm. 11, page 396 of \cite{Evans}]\label{thm_classicstrongmaximum}
 	Let $u\in C^2_1(U\times(0,T])\cap C(\bar U\times[0,T])$ be such that 
 	$$\mathrm{L}u+\frac{\partial}{\partial t}u=0,$$
 	for  $U\subset\R^n$ as above,	$\mathrm{L}+\frac{\partial}{\partial t}$  uniformly parabolic, and $\mathrm{L}$ as in \eqref{eq_parabolicdiffop}. If the maximum or minimum of $u$ is attained at a point $(x_0,t_0)\in U\times(0,T]$, then $u$ equals this value everywhere in $U\times[0,t_0]$.
 \end{theorem}
 Given a manifold $M$ with or without boundary, we set $M^{\circ}=M\setminus \partial M$, and for $x\in M$, define $M_x$ as the connected component of $M$ containing $x$. Now, the next theorem should be known, but we haven't found a reference.  
 
 \begin{theorem}\label{thm_manifoldstrongmaximum}
 	Let $(M,g)$ be an $n$-dimensional (smooth) compact Riemannian manifold with or without boundary, not necessarily connected. Suppose $u\in C^2_1(M^{\circ}\times(0,T])\cap C(M\times[0,T])$ satisfies for $(x,t)\in M^{\circ}\times(0,T] $:
 	$$ \Delta_gu(x,t)+\frac{\partial}{\partial t}u(x,t)=0.$$
 	If the maximum or minimum of $u$ is attained at a point $(x_0,t_0)\in M^{\circ}\times(0,T]$,
 	then $u$ equals this value everywhere in $ M_{x_0}\times[0,t_0]$. In particular,  the maximum and minimum of $u$ are attained in $\big(\partial M\times[0,T]\big)\cup\big( M^{\circ}\times\{0\}\big)$.
 \end{theorem}
 \begin{proof}
 	For every $\alpha\in M^{\circ}$, there is an open neighborhood  $U_\alpha\subset M$  and a chart $\boldsymbol{x}_{\alpha}:U_\alpha\ra B_\alpha\subset\R^n$, such that $\boldsymbol{x}_{\alpha}(U_\alpha)$ is an open ball $B_\alpha$, and the local representation of $\Delta_g$ in $U_\alpha$ is of type \eqref{eq_parabolicdiffop}, and satisfies \eqref{eq_uniformelliptic} for $C=1/2$. This follows from the fact that the Laplace-Beltrami operator at a point $\beta$ in the interior can be written as the usual Laplacian at $\beta$, and by continuity of the coefficients, there is an open set of $\beta$ where the inequality \eqref{eq_uniformelliptic} is true for $C=1/2$. 
 	
 	Assume there were a $t_0>0$ such that the maximum/minimum of $u$ would be attained at $(\alpha,t_0)$. Writing $\Delta_g$ w.r.t. the chart $\boldsymbol{x}_{\alpha}$ as $\Delta_\alpha$, and regarding the equation 
 	$$ \Delta_\alpha u(\boldsymbol{x}_{\alpha}^{-1}(x),t)+\frac{\partial}{\partial t}u(\boldsymbol{x}_{\alpha}^{-1}(x),t)=0,$$
 	in $B_\alpha\times(0,T]$, a neighborhood of $(\boldsymbol{x}_{\alpha}(\alpha),t_0)$, we deduce by Theorem \ref{thm_classicstrongmaximum} that  $u(x,t)\equiv u(\alpha,t_0)$ for all $(x,t)\in B_\alpha\times[0,t_0]$.
 	
 	The maximum/minimum is in particular attained at the boundary as claimed. Further, $ M_\alpha$ is covered by finitely many intersecting charts as above, and Theorem \ref{thm_classicstrongmaximum} would yield that $u$ is constant and equals $u(\alpha,t_0)$ in all of $M_\alpha\times[0,t_0]$.
  	 \end{proof}
  
   \section{The $L^2$--Norm of Gegenbauer Polynomials}\label{appendix}
 First we recall the  digamma function $\psi(x):=\tfrac{d}{dx}\log\left(\Gamma(x)\right)$ and its property:
 \begin{equation}\label{eq_digammaFormula}
 \psi(n+\tfrac{1}{2})=\sum_{k=1}^n\frac{2}{2k-1}-\gamma-\log(4),\mbox{ for }n\in\N,
 \end{equation}
 see \cite[Eq. 6.3.4]{Stegun}, where $\gamma\approx0.577$ is the Euler-Mascheroni constant.
\begin{lemma}\label{lem_gegenbauerIdentity} The Gegenbauer polynomials $\mathcal{C}_{n-2}^{(2)}(x)$ 
satisfy 
\begin{equation*}
    \int_0^1(x^2-1) \big[\mathcal{C}_{n-2}^{(2)}(x)\big]^2\dx
  =-\frac{2n^2-1}{16}\Big(\psi(n+\tfrac{1}{2})+\gamma+\log(4)\Big) +\frac{n^2}{8}.
\end{equation*}
\end{lemma}

\begin{lemma}\label{lem_gegenbauerAsymptotics} The Gegenbauer polynomials $\mathcal{C}_{n-2}^{(2)}(x)$
satisfy 
\begin{equation*}
    \int_0^1\big[\mathcal{C}_{n-2}^{(2)}(x)\big]^2\dx=\frac{n^4}{16}+\frac{4n^2-1}{64}\Big(\psi(n+\tfrac{1}{2})+\gamma+\log(4)\Big) -\frac{5}{32}n^2.
\end{equation*}
\end{lemma}
For the proofs, we need 
a result from \cite{Dette},
showing the following 
recursive formula for squares of Gegenbauer polynomials:
\begin{equation*}
\left(\frac{n}{2 \lambda}\right)^2\left[\mathcal{C}_{n}^{(\lambda)}(x)\right]^2=
\sum_{k=0}^{n-1}\frac{\lambda+k}{\lambda}\left[\mathcal{C}_{k}^{(\lambda)}(x)\right]^2-
(1-x^2)\left[\mathcal{C}_{n-1}^{(\lambda+1)}(x)\right]^2,
\end{equation*}
which, for $\lambda=1$, i.e. Chebyshev polynomials of 2nd 
kind \cite[Corollary 6.2]{Dette}, is
\begin{equation}\label{eq_Dette}
\frac{(n+1)^2}{4}\left[\mathcal{U}_{n+1}(x)\right]^2-
\sum_{k=0}^{n}(k+1)\left[\mathcal{U}_{k}(x)\right]^2=(x^2-1)
\big[\mathcal{C}_{n}^{(2)}(x)\big]^2.
\end{equation}

\begin{proof}[Proof of Lemma \ref{lem_gegenbauerIdentity}] We will use a well known identity for $m\leq n$:
\begin{equation}\label{eq_product}
  \mathcal{U}_{m}(x)\mathcal{U}_{n}(x)=\sum_{k=0}^{m}\mathcal{U}_{n-m+2k}(x),
\end{equation}
which follows by induction on $m$, starting and re-applying the recurrence \eqref{eq_ChebyshevRecurence}.
Using \eqref{eq_product} with $m=n$ in \eqref{eq_Dette} and integrating yields
  \begin{equation*}
  \begin{split}
  \int_0^1&(x^2-1) \big[\mathcal{C}_{n}^{(2)}(x)\big]^2\dx\\
  &=\frac{(n+1)^2}{4}\sum_{k=0}^{n+1}\int_0^1\mathcal{U}_{2k}(x)\dx- \sum_{k=0}^{n}(k+1)\sum_{s=0}^{k}\int_0^1\mathcal{U}_{2s}(x)\dx\\
  &= \frac{(n+1)^2}{4}\sum_{k=0}^{n+1}\frac{\mathcal{T}_{2k+1}(1)-\mathcal{T}_{2k+1}(0)}{2k+1}
- \sum_{k=0}^{n}(k+1)\sum_{s=0}^{k}\frac{\mathcal{T}_{2s+1}(1)-\mathcal{T}_{2s+1}(0)}{2s+1}\\
  &= \frac{(n+1)^2}{4}\sum_{k=0}^{n+1}\frac{1}{2k+1}- \sum_{k=0}^{n}\sum_{s=0}^{k}\frac{k+1}{2s+1},\\
  \end{split}
\end{equation*}
where we used \eqref{eq_chebyshevT_basics} and that $\mathcal{T}_{2n+1}(x)$ is odd. 
By \eqref{eq_digammaFormula}, we state for later use: 
\begin{equation}\label{eq_digamma_sum}
\int_0^1[\mathcal{U}_{n}(x)]^2\dx=\sum_{k=0}^n\frac{1}{2k+1}=\onehalf\Big(\psi(n+\tfrac{3}{2})+\gamma+\log(4)\Big),\mbox{ for }n\in\N_0.
\end{equation}
We continue
  \begin{equation*}
  \begin{split}
  \int_0^1(x^2-1) \big[\mathcal{C}_{n}^{(2)}(x)\big]^2\dx&= 
\frac{(n+1)^2}{8}\Big(\psi(n+\tfrac{5}{2})+\gamma+\log(4)\Big)\\
&\hspace{0.3cm}- \sum_{k=0}^{n}\frac{k+1}{2}\psi(k+\tfrac{3}{2})-(\gamma+\log(4))\frac{(n+2)(n+1)}{4}\\
&\hspace{-2.3cm}= \frac{(n+1)^2}{8}\psi(n+\tfrac{5}{2}) - \sum_{k=1}^{n+1}\frac{k}{2}\psi(k+\tfrac{1}{2})-\frac{(n+3)(n+1)}{8}(\gamma+\log(4)).\\
  \end{split}
\end{equation*}
Also, we find by induction:
\begin{equation*}
\sum_{k=1}^{n}\frac{k}{2}\psi(k+\tfrac{1}{2})=\frac{1}{16}\left[(2n+1)^2\psi(n+\tfrac{3}{2})-2(n+1)^2+\gamma+\log(4)\right],
\end{equation*}
where we used the recurrence 
$\psi(z+1)=\psi(z)+\tfrac{1}{z}$, see \cite[Eq. 6.3.5]{Stegun}. Thus
  \begin{equation*}
  \begin{split}
  \int_0^1(x^2-1)\big[\mathcal{C}_{n-2}^{(2)}(x)\big]^2\dx
&= \frac{2(n-1)^2-(2n-1)^2}{16}\psi(n+\tfrac{1}{2}) +\frac{n^2}{8}\\
&\hspace{2cm} - \frac{2(n+1)(n-1)+1}{16}(\gamma+\log(4))\\
  &=-\frac{2n^2-1}{16}\Big(\psi(n+\tfrac{1}{2})+\gamma+\log(4)\Big) +
  \frac{n^2}{8},
  \end{split}
\end{equation*}
finishing the proof.
\end{proof}
The proof of Lemma \ref{lem_gegenbauerAsymptotics} first needs some preparation.
\begin{lemma}\label{lem_permutation}
 Given numbers $c_{j,k}$ for $j,k\in\{0,\ldots,n\}$ such that following holds
 \begin{enumerate}
  \item $c_{j,k}=c_{j+r,k+r}$ for $j+k=n-r$ with $r\in\{1,\ldots,n\}$,
  \item $c_{j,k}=c_{n-j,k}$ for $j\geq k$,
  \item $c_{j,k}=c_{k,j}$;
 \end{enumerate}
then for any function $f:\N_0\ra\R$, we have\footnote{The apostrophe on the sum-symbol sigma means taking half the first term.}
\begin{equation}\label{eq_permutationLemma}
 \sum_{j,k=0}^{n}c_{j,k}\cdot f(|j-k|)=\sum_{j,k=0}^{n}c_{j,k}\cdot f(|n-j-k|)= 2\sideset{}{'}\sum_{r=0}^{n}f(r)\sum_{u=0}^{n-r}c_{r+u,u}.
\end{equation}
\end{lemma}
\begin{proof}
We first fix some $r\in\{1,\ldots,n\}$ and regard the second sum.
Observe that all tuples 
that satisfy $j_i+k_i=n-r$ and $\hat j_i+\hat k_i=n+r$ yield $|n-j-k|=r$, and are listed:
\begin{equation*}
\begin{array}{c|c|c|c|c}
 i & 1 & 2 & \ldots   & n-r+1\\
 \hline
 j_i & 0 & 1 & \ & n-r\\
 k_i & n-r & n-r-1 & \ldots  & 0\\
 \hat j_i & r & r+1 & \ & n\\
 \hat k_i & n & n-1 & \  & r\\
\end{array}
\end{equation*}
So for all $r$, $(j,k)\mapsto (j+r,k+r)=:(\hat j_i,\hat k_i)$ is a bijection with $c_{j_i,k_i}= c_{\hat j_i,\hat k_i}$ and
$$ \sum_{j,k=0}^{n}c_{j,k}\cdot f(|n-j-k|)=2 \sum_{\substack{j,k=0 \\j+k<n}}^{n}c_{j,k}\cdot f(n-j-k)+f(0)\cdot\sum_{u=0}^n c_{n-u,u}.$$
The first sum of \eqref{eq_permutationLemma} can be restricted to $j>k$ when doubled, apart of the sum  $f(0)\cdot\sum_{u=0}^nc_{u,u}$. Again, we list all tuples with $j_i-k_i=r=n-\hat j_i-\hat k_i$:
\begin{equation*}
\begin{array}{c|c|c|c|c}
 i & 1 & 2 & \ldots   & n-r+1\\
 \hline
 j_i & r & r+1 & \ & n\\
 k_i & 0 & 1 & \ldots  & n-r\\
 \hat j_i & n-r & n-r-1 & \ & 0\\
 \hat k_i & 0 & 1 & \  & n-r\\
\end{array}
\end{equation*}
Similarly, $(j,k)\mapsto (n-j,k)=:(\hat j_i,\hat k_i)$ is a bijection with $c_{j_i,k_i}= c_{\hat j_i,\hat k_i}$, and
$$\sum_{j>k=0}^{n}c_{j,k}\cdot f(j-k)=\ \sum_{\substack{j,k=0 \\j+k<n}}^{n}c_{j,k}\cdot f(n-j-k).$$ 
Rewriting the first sum above via $j=r+u$ and $k=u$ for some  $u\in\{0,\ldots,n-r\}$ and using that $c_{n-u,u}=c_{u,u}$ finishes the argument.
\end{proof}
 Requirement $2.$ in Lemma \ref{lem_permutation} is valid  for all $j,k$. To see this, let $j<k$, then
 $$c_{j,k}\stackrel{2.+3.}{=}c_{j,n-k}\stackrel{1.}{=}c_{j+(k-j),n-k+(k-j)}=c_{k,n-j}\stackrel{3.}{=}c_{n-j,k}. $$

 \begin{lemma}\label{lem_GegenbauerCosineSquared} 
 Let $n,\lambda\in\N$ be fixed, $j,k\in\{0,\ldots,n\}$ and define
 $$c_{j,k}^\lambda=c_{j,k}^\lambda(n)=\frac{1}{[\Gamma(\lambda)]^4}\frac{\Gamma(\lambda+j)\Gamma(\lambda+n-j)}
 {j!(n-j)!}\frac{\Gamma(\lambda+k)\Gamma(\lambda+n-k)}
 {k!(n-k)!},$$
 then\footnote{The apostrophe on the sum-symbol sigma means taking half the first term.}
  $$\big[\mathcal{C}_{n}^{(\lambda)}\big(\cos(t)\big)\big]^2=2\sideset{}{'}\sum_{r=0}^n\cos(2 r t)\sum_{u=0}^{n-r}c_{r+u,u}^\lambda.$$
 \end{lemma}
 \begin{proof}
  We will use Lemma \ref{lem_permutation} with \cite[Eq. 8.934]{Gradshteyn}:
\begin{equation}\label{eq_GegenbauerCosine}
\mathcal{C}_{n}^{(\lambda)}(\cos(\varphi))=\sum_{\substack{k,\ell=0 \\ k+\ell=n }}^{n}\frac{\Gamma(\lambda+k)\Gamma(\lambda+\ell)}
 {k!\ell![\Gamma(\lambda)]^2}\cos((k-\ell)\varphi),
\end{equation}
  in conjunction with the angle-sum and half-angle formula for cosine and sine:
 \begin{equation*}
  \begin{split}
  \big[\mathcal{C}_{n}^{(\lambda)}(1)\big]^2&-
 \big[\mathcal{C}_{n}^{(\lambda)}\big(\cos(t)\big)\big]^2= \sum_{j,k=0}^{n}c_{j,k}^\lambda\Big(1-\cos\big((n-2j)t \big)\cos\big((n-2k)t \big)\Big)  \\
 &= \sum_{j,k=0}^{n}c_{j,k}^\lambda\onehalf\Big(1-\cos\big((j-k)2t \big)+1-\cos\big((n-j-k)2t \big)\Big) \\
 &= \sum_{j,k=0}^{n}c_{j,k}^\lambda\Big(\sin\big((j-k)t\big)^2+\sin\big((n-j-k)t\big)^2\Big) \\
 &= 4\sum_{r=1}^{n}\sin\big(rt\big)^2\cdot\sum_{u=0}^{n-r}c_{r+u,u}^\lambda. \\
  \end{split}
 \end{equation*}
 Hence 
  \begin{equation*}
  \begin{split}
  \big[\mathcal{C}_{n}^{(\lambda)}(1)\big]^2&-\left(\big[\mathcal{C}_{n}^{(\lambda)}(1)\big]^2-
 \big[\mathcal{C}_{n}^{(\lambda)}\big(\cos(t)\big)\big]^2\right)\\
 &\hspace{1cm}=2\sideset{}{'}\sum_{r=0}^{n}\sum_{u=0}^{n-r}c_{r+u,u}^\lambda- 4\sum_{r=1}^{n}\sin\big(rt\big)^2\cdot\sum_{u=0}^{n-r}c_{r+u,u}^\lambda \\
  &\hspace{1cm}=\sum_{u=0}^{n}c_{u,u}^\lambda+2\sum_{r=1}^{n}\big(1-2\sin\big(rt\big)^2\big)\cdot\sum_{u=0}^{n-r}c_{r+u,u}^\lambda, \\
   \end{split}
 \end{equation*}
 and we finish using $1-2\sin(rt)^2=\cos(2rt)$.
 \end{proof}

 \begin{proof}[Proof of Lemma \ref{lem_gegenbauerAsymptotics}] With notation of Lemma \ref{lem_GegenbauerCosineSquared} and $c_{j,k}=c_{j,k}^2(n-2)$ 
   \begin{equation*}
 \begin{split}
 & \sum_{u=0}^{n-2-r} c_{r+u,u}
  =\sum_{u=1}^{n-1-r}(r+u)(n-u)u(n-r-u)\\
  &=\tfrac{4r^2-1}{120}\Big( r\big(5n^2-\tfrac{1}{4}\big)-r^3-5n(n^2-1) \Big)-
\frac{2r}{64}(4n^2-1)+\binom{2n+2}{5}\frac{1}{8},\\
 \end{split}
\end{equation*}
 and by Lemma \ref{lem_GegenbauerCosineSquared}
 \begin{equation*}
\begin{split}
  \int_{0}^1 \big[\mathcal{C}_{n-2}^{(2)}(x)\big]^2\dx&=\sum_{u=0}^{n-2}c_{u,u}+2\sum_{r=1}^{n-2}\int_0^{\frac{\pi}{2}}\cos(2rt)\sin(t)\dt\sum_{u=0}^{n-2-r}c_{r+u,u}\\
  &=\sum_{u=0}^{n-2}c_{u,u}-2\sum_{r=1}^{n-2}\frac{1}{4r^2-1}\sum_{u=0}^{n-2-r}c_{r+u,u}\\
  &=\frac{n^5-n}{30}-\frac{1}{60}\sum_{r=1}^{n-2}\Big( r\big(5n^2-\tfrac{1}{4}\big)-r^3-5n(n^2-1) \Big)\\
  &\hspace{1cm}+2\sum_{r=1}^{n-2}\frac{1}{4r^2-1}
  \bigg(\frac{2r}{64}(4n^2-1)-\binom{2n+2}{5}\frac{1}{8}\bigg)\\
  &=\frac{2n^4-5n^2}{32}+\sum_{r=0}^{n-1}\frac{1}{4r^2-1}\frac{2r-1}{32}(4n^2-1).\\
  \end{split}
  \end{equation*}
  An application of \eqref{eq_digamma_sum} finishes the proof.
\end{proof}
\section{Sampling on $\SO$}
So far we obtained theoretical bounds for the Green energy on $\SO$ via a lemma due to N. Elkies and properties of points sampled by a dpp. The upper bound cannot be best possible, as it is an expected value -- and hence there must be fluctuations above and in particular below that value.

In this section we will introduce an algorithm to sample points in $\SO$, that is simple to implement and numerically outperforms points sampled by a dpp. 
We are not giving any proofs regarding this algorithm, but rather show that it exists and how our bounds could be used as a comparison tool.

 In 1987 a probabilistic algorithm was introduced by P. Diaconis and M. Shahshahani for compact groups in \cite{Diakonis} and seemingly a special case of that was re-discovered by J. Arvo for $\SO$ in \cite{Arvo}. 
	 We will use a variant of this, replacing random points by a Halton sequence in the unit cube, which we baptize HArDiSh algorithm, and it does very well according to numerics. See Figure \ref{fig:GreenEnergy}.

Following closely to \cite{Arvo}, we sample $N$ points as follows: For $x_1,x_2,x_3$ to be determined later, let $M=-HR$ where $H=\fatone-2vv^t$, 
\begin{equation}\label{eq_hardish}
	v=\frac{1}{\sqrt{N}}\begin{pmatrix}
		\cos(2\pi x_2)\sqrt{x_3}\\
		\sin(2\pi x_2)\sqrt{x_3}\\
		\sqrt{N-x_3}
	\end{pmatrix},
	\mbox{ and }
		R=\begin{pmatrix}
	\cos(2\pi x_1) &\sin(2\pi x_1) & 0\\
	-\sin(2\pi x_1)&\cos(2\pi x_1) &0\\
	0&0&1
	\end{pmatrix}.
\end{equation}
In \cite{Arvo}, the $x_j$ were chosen uniformly at random, and as Arvo already mentions, generating $x_j$ by stratified or jittered sampling should yields less clumping for the matrices $M$. Our humble modification is to sample $x_j$ via Halton sequences, i.e. let vdC$(p,j)$ denote the $j$-th element of the van der Corput seqence in base $p$, set
\[
	\mathcal{H}=\begin{pmatrix}
1/3&2/3 &1/9&4/9&7/9&\ldots&\mbox{vdC}(3,N)\\
1/2 & 1/4 & 3/4&1/8 &5/8&\ldots&\mbox{vdC}(2,N)\\
1&2&3&4&5&\ldots&N
\end{pmatrix};
\]
then we obtain matrices $\{M_k\}_{k=1}^{N}$ via \eqref{eq_hardish} by setting $x_j(k)=\mathcal{H}(j,k)$. We do not know if the algorithm will continue to perform well for high numbers $N$.
\begin{figure}[h]
 \centering
 \includegraphics[width=11cm,keepaspectratio=true]{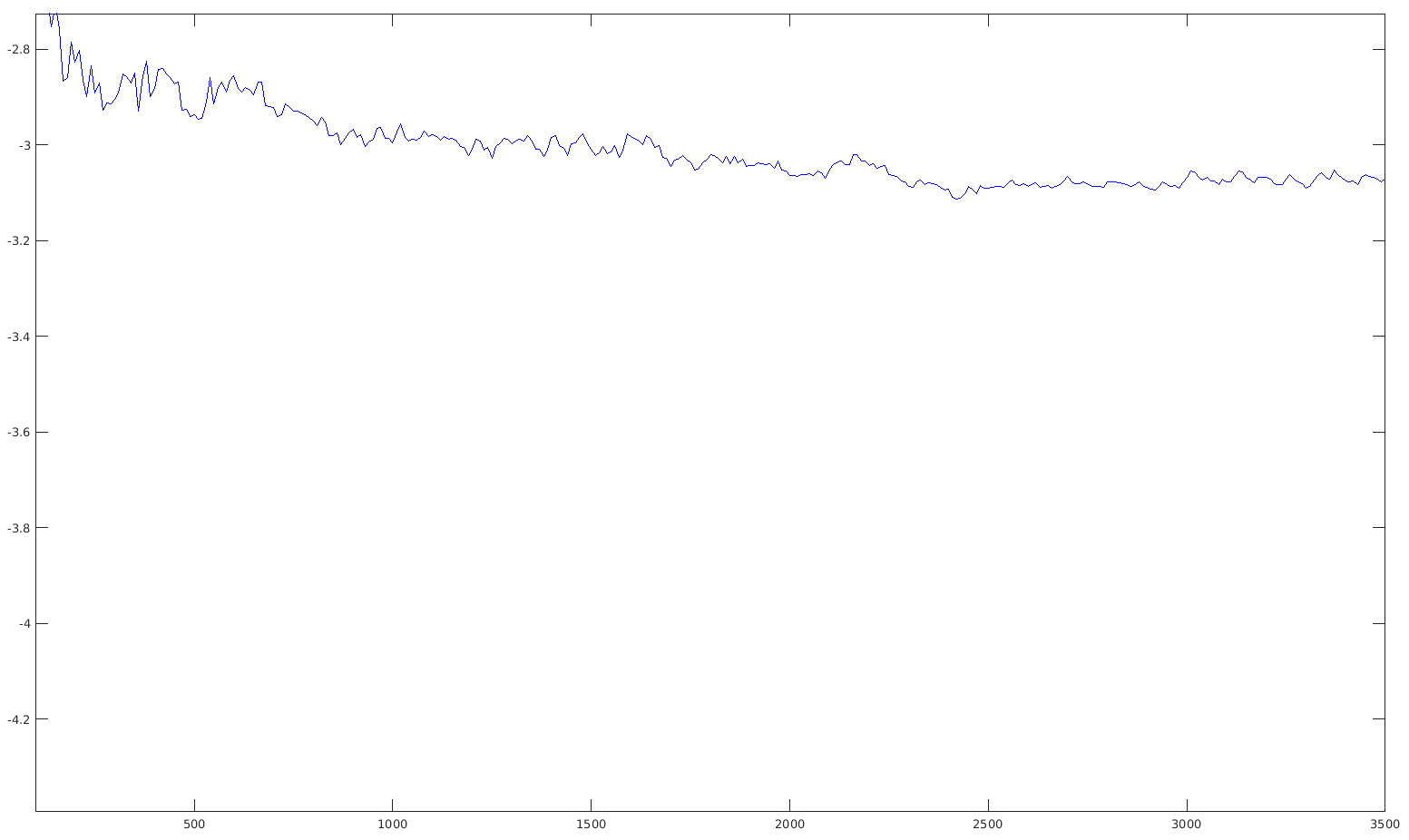}
  \caption{The graphic shows the evolution of the Green energy divided by $N^{3/4}$ for HArDiSh -- generated points, here $N=k*10$ for $k\in\{10,\ldots,350\}$. The boundaries for the y-axis  are chosen to be our theoretical bounds.}
  \label{fig:GreenEnergy}
\end{figure}


\begin{thebibliography}{10} 


\bibitem{Stegun} M. Abramowitz, I. A. Stegun: \emph{Handbook of Mathematical Functions With Formulas, Graphs, and Mathematical Tables}.
Department of Commerce (USA), National Bureau of Standards, Applied Mathematics Series 55 (1972).

\bibitem{ASZ14} K.~Alishashi, M.~S. Zamani: \emph{The spherical ensemble and uniform 
	distribution of points on the sphere}. Electron. J. Probab. 20, no. 23, 
27 pp (2015).

\bibitem{Arvo} J. Arvo: \emph{Fast Random Rotation Matrices}. Graphics Gems III (Edited by David Kirk), AP Professional (2012).

\bibitem{Aubin} T. Aubin: \emph{Some Nonlinear Problems in Riemannian Geometry}. Springer
Monographs in Mathematics, Springer Berlin Heidelberg (1998).

\bibitem{Juan} C. Beltrán, J. G. Criado del Rey and N. Corral: \emph{Discrete and Continuous Green Energy on Compact Manifolds}. J. Approx. Theory 237, pp 160–185 (2019).


\bibitem{Etayo} C. Beltrán and U. Etayo: \emph{The Projective Ensamble and Distribution of Points in 
	Odd-Dimensional Spheres}. Constr. Approx.  48, no. 1, pp 163–182 (2018).


\bibitem{Beltran} C. Beltrán, J. Marzo and J. Ortega-Cerdà: \emph{Energy and discrepancy of rotationally
invariant determinantal point processes in high dimensional spheres}. J. Complexity 37, pp 76-109 (2016).

\bibitem{GAF} J. Ben Hough, M. Krishnapur, Y. Peres, V. Virág: \emph{Zeros of Gaussian Analytic
	Functions and Determinantal Point Processes}. American Mathematical Society, Providence, RI (2009). 

\bibitem{Criado} J. G. Criado del Rey: \emph{On the Separation Distance of Minimal Green Energy Points on Compact Riemannian Manifolds}. arXiv:1901.00779v1 (2019).


\bibitem{Diakonis} P. Diaconis and M. Shahshahani: \emph{The subgroup algorithm for generating uniform random variables}. Prob. Eng. Inf. Sc. 1, pp 15–32  (1987).

\bibitem{DoCarmo} M. P. do Carmo: \emph{Riemannian Geometry}. Birkh\"{a}user Boston (1992).


\bibitem{Dette} H. Dette: \emph{New identities for orthogonal polynomials on a compact interval}.
J. Math. Anal. Appl. 179, pp 547-573  (1993). 


\bibitem{Engo} K. Engø: \emph{On the BCH-Formula in so(3)}. BIT Numerical Mathematics 41, pp 629-632, https://doi.org/10.1023/A:1021979515229 (2001).

	\bibitem{Evans} L. C. Evans: \emph{Partial Differential Equations}. American Mathematical Society, 2nd Edition (2010).
	
\bibitem{Gradshteyn} I. S. Gradshteyn, I. M. Ryzhik, A. Jeffrey, D. Zwillinger: \emph{Table of Integrals, Series, and Products}.
Academic Press; 6th edition (2000).

\bibitem{Schmid} T. Hangelbroek, D. Schmid: \emph{Surface Spline Approximation on SO(3)}. 
Appl. Comput. Harmon. Anal. Volume 31, Issue 2, pp 169-184 (2011).

\bibitem{Huynh} Du Q. Huynh: \emph{Metrics for 3D Rotations: Comparison and Analysis}.
J Math Imaging Vis 35, pp 155-164  (2009). 

\bibitem{Joshi} A. W. Joshi: \emph{Elements of Group Theory for Physicists}. Wiley Eastern Private Limited, New Delhi, (1973).

\bibitem{Jost} J. Jost: \emph{Riemannian geometry and geometric analysis} {(Sixth edition)}.
Universitext, Springer, Heidelberg (2011).


\bibitem{Lang} S. Lang: \emph{Introduction to Arakelov Theory}. Springer-Verlag New York (1988).

\bibitem{Novelia} A. Novelia and O. M. O'Reilley: \emph{On Geodesics of the Rotation Group SO(3)}.
Regular and Chaotic Dynamics, Vol. 20, No. 6, pp 729–738. Pleiades Publishing, Ltd., (2015).

\bibitem{Marzo} J. Marzo and J. Ortega-Cerdà: \emph{Expected Riesz energy of some determinantal 
processes on flat tori}. Constructive Approximation 47 (1), pp 75-88 (2018).



\bibitem{Rider} B. Rider and B. Virág: \emph{Complex determinantal processes and $H^1$ noise}. 
Electronic Journal of Probability 12, pp 1238-57 (2007).








%


\bibitem{Shub} M. Shub and S. Smale: \emph{Complexity of Bezout's theorem II -- Volumes and
probabilities}. Computational algebraic geometry, pp 267–285,
Progr. Math., 109, Birkhäuser Boston, Boston, MA (1993).

\bibitem{Smale} S. Smale: \emph{Mathematical Problems for the Next Century}. Math. Intelligencer 20, No. 2, pp 7-15, (1998).


 \bibitem{Szego} G. Szeg\"{o}: \emph{Orthogonal Polynomials}. Amer. Math. Soc. (1939).


\bibitem{Vollrath} A. Vollrath: \emph{The Nonequispaced Fast SO(3) Fourier Transform,
	Generalisations and Applications} {(PhD-Thesis)}. University of Lübeck (2010).


\bibitem{Wigner} E. P. Wigner: \emph{Group Theory and its Application to the Quantum Mechanics of Atomic Spectra}. Academic Press, New York and London, (1959).




\end{thebibliography}
\end{document}